\documentclass[12pt,onecolumn,english]{IEEEtran}             
\IEEEoverridecommandlockouts           
\usepackage{epsf}
\usepackage{epsfig}
\usepackage{times}
\usepackage{amssymb}
\usepackage{amsmath}
\usepackage{amsfonts}
\usepackage{latexsym}
\usepackage{url}
\usepackage{subcaption}
\usepackage{tikz}
\usepackage{dblfloatfix} 
\usepackage{blkarray}
\usepackage{hyperref}
\hypersetup{
    colorlinks,
    linktocpage,
    linkcolor=blue,
}

\usetikzlibrary{calc,shadings,patterns}

\newcommand{\markov}{\mathrel\multimap\joinrel\mathrel-\mspace{-9mu}\joinrel\mathrel-}
\newtheorem{theorem}{Theorem}

\newtheorem{corollary}{Corollary}
\newtheorem{example}{Example}
\newtheorem{remark}{Remark}
\newtheorem{proof}{Proof}
\newtheorem{conjecture}{Conjecture}

\title{Diamond Message Set Groupcasting: From an Inner Bound for the DM Broadcast Channel to the Capacity Region of the Combination Network} 

\author{Mohamed Salman and Mahesh K. Varanasi \\
\thanks{This work was presented in part at the 2020 IEEE International Symposium on Information Theory, LA, CA \cite{salman2020diamond}. This research was funded in part by the 2018 and 2019 Qualcomm Faculty Awards. M. Salman and M. K. Varanasi are with the Electrical, Computer and Energy Engineering Department, University of Colorado, Boulder, CO, USA (emails: \{mohamed.salman, varanasi\}@colorado.edu).}
}

\begin{document}

\maketitle
\thispagestyle{plain}
\pagestyle{plain}
\pagenumbering{arabic}

\begin{abstract}
Multiple groupcasting over the broadcast channel (BC) is studied in a special setting. In particular, an inner bound is obtained for the $K$-receiver discrete memoryless (DM) BC for the diamond message set which consists of four groupcast messages: one desired by all receivers, one by all but two receivers, and two more desired by all but each one of those two receivers. The inner bound is based on rate-splitting and superposition coding  and is given in explicit form herein as a union over coding distributions of four-dimensional polytopes. When specialized to the so-called combination network, which is a class of three-layer (two-hop) broadcast networks parameterized by $2^K-1$ finite-and-arbitrary-capacity noiseless links from the source node in the first layer to as many nodes of the second layer, our top-down approach from the DM BC to the combination network yields an explicit inner bound as a single polytope via the identification of a single coding distribution. This inner bound consists of inequalities which, in a problem that is akin to finding a few needles in a haystack, are then identified to be within the class of a plethora of (indeed, infinitely many) generalized cut-set outer bounds recently obtained by Salimi et al for broadcast networks. We hence establish the capacity region of the general $K$-user combination network for the diamond message set, and do so in explicit form. Such a result implies a certain strength of our inner bound for the DM BC in that it (a) produces a hitherto unknown capacity region when specialized to the combination network and (b) may capture many combinatorial aspects of the capacity region of the $K$-receiver DM BC  itself (for the diamond message set). Moreover, we further extend that inner bound by adding binning to it and providing that inner bound also in explicit form as a union over coding distributions of four-dimensional polytopes in the message rates.
\end{abstract}


\section{Introduction}
\label{Sec_Into}

A general order-theoretic framework for groupcasting over the $K$-receiver DM BC was proposed 
in \cite[Theorem 1]{romero2016superposition}, that allows for a succinct albeit indirect description of the rate region achieved by {\em up-set} rate-splitting and superposition coding 
for simultaneously sending the complete set of $2^K{-}1$ independent messages, each desired by some distinct subset of receivers. In up-set rate-splitting, a message intended for some subset of receivers is split into sub-messages, with each sub-message to be delivered to some distinct subset of receivers for all possible such subsets that include the originally intended set of receivers. The set of sub-messages that are intended for the same group of receivers are then collected to form a new reconstructed message. 
Among the numerous choices \cite{romero2017unifying}, the type of superposition coding used for generating the codebooks for the reconstructed messages is the one with the superposition order taken to the subset inclusion order \cite{romero2016superposition}. Finally, each receiver jointly decodes its desired reconstructed messages which contain the desired messages as well as the partial interference contained in the undesired sub-messages assigned to it via rate-splitting. 

The inner bound in \cite[Theorem 1]{romero2016superposition} is given in implicit form in that it gives an indirect description of per-distribution polytopes in terms of the split rates (in higher dimension) rather than the original messages rates. In principle, the split rates can be projected away with Fourier Motzkin elimination (FME) \cite{schrijver1998theory}, but in practice, this is only possible for small settings, i.e., small number of receivers and/or small message sets since projecting away all split rates in this general setting is intractable in general.

In a different line of work on network coding, certain three-layered broadcast (single-source) networks named combination networks were introduced in \cite{ngai2004network} to demonstrate that network coding for a single multicast session can attain unbounded gain over routing alone. A combination network is defined more generally in \cite{salimi2015generalized} to be a class of three-layer (two-hop) broadcast networks (as depicted in Fig. \ref{Fig_K3_combination networks} for $K=3$) parameterized by $2^K{-}1$ finite-and-arbitrary-capacity noiseless links from the source node in the first layer to as many nodes of the second layer, with each node in the second layer connected to a distinct subset of $K$ destination nodes via infinite capacity links. For the complete message set 
linear network coding schemes and matching converses were provided to obtain the capacity region of the two and three-receiver combination network in \cite{grokop2008fundamental} (see also \cite{salimi2015generalized}) and the symmetric capacity region (in which the rates of the messages desired by the same number of users are equal) of the special class of symmetric combination networks, in which the capacity of links to nodes in the second layer that are connected to the same number of receivers are identical, in \cite{tian2011latent,salimi2015generalized}. The general 
capacity region of the general 
combination network for a complete message set (with $2^K{-}1$ messages) for $K{>}3$ remains an open problem. In this regard, the work of \cite{bidokhti2016capacity} must be mentioned for having produced partial results: in particular, for two nested messages, i.e., a multicast message intended for all receivers and a private message intended for a subset of the receivers, the capacity of the $K$-user general combination network is established in \cite{bidokhti2016capacity} when the number of receivers that demand only the multicast message is no more than three.

The general combination network can be seen as a special class of deterministic DM BCs as observed in \cite{romero2016superposition}. A top-down approach was initiated therein to study the combination network. In particular, a portion (i.e., an inner bound) of the general inner bound of \cite[Theorem 1]{romero2016superposition} based on up-set rate-splitting and superposition coding for the DM BC for a complete message set is proposed as an achievable rate region for the combination network in \cite[Theorem 2]{romero2016superposition} through the specification of a single random coding distribution. Its indirect description notwithstanding, it was shown in \cite{romero2016superposition}, via its alternative description in \cite[Theorem 2]{romero2016superposition}, that that achievable rate region recovers the capacity (for $K=3$) and symmetric capacity (for general $K$) for symmetric combination networks, results previously obtained via linear network coding in \cite{grokop2008fundamental,tian2011latent,salimi2015generalized}. Hence, the work in \cite{romero2016superposition} provides a certain validation for the strength of the rate-splitting and superposition coding inner bound of
\cite[Theorem 1]{romero2016superposition} for the $K$-receiver DM BC itself for the complete message set. 

Bolstered by the success of the top-down approach of 
\cite{romero2016superposition}, we continue that study but for general (not only symmetric) combination networks for any $K>3$ but by focusing on the diamond message set that consists of four messages, one desired by all $K$ receivers, one by all but two receivers, and two more desired by all but each one of those two receivers\footnote{The name diamond message set is derived from the depiction of the associated ordered set of four message indices, each message index being a subset of the set of all receiver indices at which that message is desired, with the order taken to be the subset inclusion order, in the form of a Hasse diagram, which would hence be diamond-shaped.}. In particular, we show that, adopting the general framework of \cite[Theorem 1]{romero2016superposition} as is, even in this incomplete message set case\footnote{A more general framework that builds on that of \cite{romero2016superposition} is given in \cite{romero2017rate} to account for incomplete message sets in general, but that extra generality turns out to be unnecessary for the sake of discovering the capacity of the combination network for the diamond message set, as demonstrated here.}
is sufficient to produce the hitherto unknown capacity region of the combination network for the diamond message set.
The capacity result for the combination network in the special case of three degraded messages, obtained by setting the rate of one of the messages desired by all but one receiver to zero, was also previously unknown.
Moreover, special instances of the diamond message set capacity result for various two message set cases, by setting the rates of the other two messages to zero, recover two-message capacity results for the combination network from \cite{bidokhti2016capacity} and \cite{salman2018achievable}. 

 In particular, 
 we present an inner bound for $K$-receiver DM BC in terms of the actual rates of the four messages instead of the split rates as in \cite[Theorem 1]{romero2016superposition} by projecting away the split rates via Fourier Motzkin elimination \cite{schrijver1986theory} as described in Appendix \ref{Appendix_FME_TH1}, in spite of the indeterminate number of inequalities that describe the per-distribution achievable rate region in original- and split-rate space. We then specialize this inner bound to the combination network by choosing a single distribution for the auxiliary and input random variables to obtain an explicit polyhedral inner bound for the general combination network. 

We establish the capacity result by identifying
the inequalities present in the aforementioned explicit inner bound to be within the class of infinitely many possible generalized cut set bounds proposed in \cite{salimi2015generalized} as outer bounds for broadcast networks based on the sub-modularity of entropy. Notably, our order-theoretic description of the inner bound enables that identification by permitting a description of the inequalities in the inner bound in a form that is similar to the form of the generalized cut-set bounds of \cite{salimi2015generalized}. 

Beyond establishing the capacity of the general $K$-user combination network for the diamond message set, our top-down approach suggests that the inner bound of \cite[Theorem 1]{romero2016superposition} given in a more explicit form here is a good one for the much more general DM BC (with the diamond message set) in that it specializes to the capacity region in the combination network without symmetry assumptions, and possibly captures many combinatorial aspects of the capacity region of the DM BC itself. Nevertheless, we extend that inner bound by adding binning to rate-splitting and superposition coding and provide an expression for the resulting inner bound also in explicit form as a union over coding distributions of four-dimensional polytopes in the message rates. 




The rest of this paper is organized as follows. In Section \ref{Sec:System}, we present the order theory notation and describe the system model. In Section \ref{Sec_Main_results}, we present the inner bound for the DM BC and its specialization to the general $K$-user combination network along with a proof of the converse, establishing the latter's capacity region 
for the diamond message set. In Section \ref{sec:binning}, we extend the inner bound for the DM BC of Section \ref{Sec_Main_results} based on rate-splitting and superposition coding by adding binning to it. While binning is not needed to achieve the capacity of the combination network it is an interesting open question as to whether the more general inner bound of Section \ref{sec:binning} would be optimal for BCs that are more general than the combination network such as the deterministic broadcast channel.
Finally, the paper is concluded in Section \ref{Sec_Conc}.  An outline of the Fourier-Motzkin Elimination used to obtain the inner bound for the DM BC of Section \ref{Sec_Main_results} in terms of explicit per-distribution polytopes is provided in Appendix \ref{Appendix_FME_TH1} and an outline of the proof of the more general inner bound resulting from the inclusion of binning is given in Appendix \ref{app:thmbinning}.

\section{Notation and System Model}
\label{Sec:System}

\subsection{Order Theory}
We find it convenient to introduce ideas from order theory following the notation in \cite{romero2016superposition} to enable the descriptions of the system model and the results. We consider the ground set to be an ordered set of subsets of receiver indices in $\{1,2, \cdots , K\} \triangleq [1:K]$. We assume the order to be that of set inclusion, i.e., $S\leq S^{'}$ if and only if $S \subseteq S^{'}$ while $S$ and $S^{'}$ are incomparable
if neither $S \subseteq S^{'} $ nor $S^{'} \subseteq S$. Let $\mathsf{P}$ be such an ordered set of sets and $\mathsf{Q}$ be a subset of $\mathsf{P}$. Note that the sans serif letter is used to represent a set of sets like $\mathsf{P},\mathsf{Q}$ and $\mathsf{E}$ to distinguish it from sets.  
We say that $\mathsf{Q}$ is an {\em up-set} if $S\in \mathsf{Q}$, $S^{'}\in \mathsf{P}$, and $S^{'} \geq S$ implies $S^{'}\in Q$ and a {\em down-set} if $S\in Q$, $S^{'}\in \mathsf{P}$, and $S^{'}\leq S$ implies $S^{'}\in \mathsf{Q}$.

Moreover, for any subset $\mathsf{Q}\subseteq \mathsf{P}$, we define the smallest {\em down-set} containing $\mathsf{Q}$ as $\downarrow_{\mathsf{P}} \mathsf{Q} = \{S^{'} \in \mathsf{P}: S^{'}\leq S, S \in \mathsf{Q}\}$ and the smallest {\em up-set} containing $\mathsf{Q}$ as $\uparrow_{\mathsf{P}} \mathsf{Q} = \{S^{'} \in \mathsf{P}: S\leq S^{'}, S \in \mathsf{Q}\}$. For the sake of simplicity, we abbreviate the set $\{i_1,i_2,..,i_N\} \subseteq \{1, \cdots , K\}$ for any positive number $N\leq K$ as $i_1i_2 \cdots i_N$, 
adopting the convention that $i_1 < i_2 < \cdots < i_N$. For instance, let $\mathsf{P}$ be the power set of $[1:3]$, i.e., $\mathsf{P} = \{\phi, 1,2,3,12,13,23,123\}$. Then, for instance, we have $\downarrow_{\mathsf{P}} \{13\}=\{\phi,1,3,13\}$ and $\uparrow_{\mathsf{P}} \{13\}= \{13,123\}$. In some cases, especially when the set $S=\{i_1,i_2,..,i_N\}$ has many elements, we find it more convenient to denote it by its complement, i.e., $\overline{S}=[1:K]\backslash S$. For example, the set $\overline{\{i_1\}}$, also denoted $\overline{i}_1 $, and is the set $ [1:K]\backslash \{i_1\}$.

Finally, for any $\mathsf{F}\subseteq \mathsf{P}$ and $i\in[1:K]$, we define $\mathsf{W}_i^{\mathsf{F}}$ as the set of sets in $\mathsf{F}$ containing $i$ so that $ \mathsf{W}_i^{\mathsf{F}} \triangleq \{S\in \mathsf{F}:i\in S\}$.

From the given notation, we can show that the following relationships are true:
\begin{enumerate}
\item For any set $S=i_1i_2\cdots i_N \subseteq \{1,2,\cdots ,K\}$, we have  
\begin{align}
\cup_{k\in S}\mathsf{W}_k^{\mathsf{P}}&= \uparrow_{\mathsf{P}}\{i_1,i_2,\cdots, i_N\}\label{Eq_Lemma_1}\\
\cap_{k\in S}\mathsf{W}_k^{\mathsf{P}}&= \uparrow_{\mathsf{P}}\{i_1i_2\cdots i_N\} \label{Eq_Lemma_2}
\end{align}
\item For any set $S=i_1i_2\cdots i_N \subset \{1,2,\cdots, K\}$ and any $i\in [1:K]$ 
\begin{align}
\downarrow_{\mathsf{W}_i^{\mathsf{P}}}\{\overline{i_1},\overline{i_2},\cdots, \overline{i_N}\} \cup \uparrow_{\mathsf{W}_i^\mathsf{P}}\{S\}&=\mathsf{W}_i^{\mathsf{P}} \label{Eq_Lemma_3}\\
\downarrow_{\mathsf{W}_i^{\mathsf{P}}}\{\overline{i_1},\overline{i_2},\cdots, \overline{i_N}\} \cap \uparrow_{\mathsf{W}_i^\mathsf{P}}\{S\}&=\phi \label{Eq_Lemma_4}\\
\downarrow_{\mathsf{W}_i^{\mathsf{P}}}\{\overline{S}\} \cup \uparrow_{\mathsf{W}_i^\mathsf{P}}\{i_1,i_2,\cdots ,i_N\}&=\mathsf{W}_i^{\mathsf{P}} \label{Eq_Lemma_5}\\
\downarrow_{\mathsf{W}_i^{\mathsf{P}}}\{\overline{S}\} \cap \uparrow_{\mathsf{W}_i^\mathsf{P}}\{i_1,i_2,\cdots ,i_N\}&=\phi
\label{Eq_Lemma_6}
\end{align}

\end{enumerate}

\subsection{System Model}
We consider a DM BC with the transmitter denoted by the transmitted symbol $X\in \mathcal{X}$, $K$ receivers denoted by their respective channel outputs $Y_i\in \mathcal{Y}_i$, and the channel transition probability $W(y_1y_2 \cdots y_K|x)$ where the conditional probability of $n$ channel outputs $Y_1^n, \cdots , Y_K^n $ with $Y_k^n \triangleq (Y_{k,1}, \cdots , Y_{k,n})$, conditioned on $n$ channel inputs $X^n \triangleq (X_{1}, \cdots , X_n)$  is given by $$p(y_1^n\cdots y_K^n|x^n)= \prod_{j=1}^nW(y_{1j}\cdots y_{Kj}|x_j)$$ where $X_j, Y_{1,j},\cdots Y_{K,j}$ are the channel input and outputs in the $j^{th}$ channel use. 

The message $M_{S} \in [1:2^{nR_S}] $ of rate $R_S$ is indexed by the subset $S\subseteq [1:K]$ of receivers it is intended for. Hence, $M_{\overline{S}}$ is the message intended for all receivers except the receivers in $\overline{S}$. The diamond message set consists of four messages, $M_{\overline{\phi}}$, a message intended for all receivers, $M_{\overline{K-1.K}}$ a message intended for the first $K-2$ receivers\footnote{The dot separating $K-1$ and $K$ in $M_{\overline{K-1.K}}$ is a slight abuse of notation since the subscript in $M_{\overline{K-1K}}$ can be confusing.}, and $M_{\overline{K-1}}$ and $M_{\overline{K}}$, two messages intended for all but receivers $K-1$ or $K$. Define $\mathsf{E}$ as the set of all message indices so that  $\mathsf{E}=\{\overline{\phi},\overline{K},\overline{K-1},\overline{K-1.K}\}$. Note that receiver $Y_K$ demands $M_{\overline{\phi}}$ and $M_{\overline{K-1}}$, receiver $Y_{K-1}$ demands $M_{\overline{\phi}}$ and $M_{\overline{K}}$, while the rest of receivers $\{Y_i\}_{i=1}^{K-2}$ demand all four messages. 

The combination network \cite{salimi2015generalized,romero2016superposition} is a special case of the general DM-BC. As shown in Fig. \ref{Fig_K3_combination networks} for the three-receiver case, it consists of three layers of nodes. The top layer and bottom layer
 consist of the single source node $X$ and $K$ receivers $\{Y_i\}_{i=1}^K$, respectively. While, the middle layer consists of $2^K-1$ intermediate nodes, denoted $V_{S}$ for all $S\in \mathsf{P}$ where $\mathsf{P}$ be the power set of $[1:K]$ excluding the empty set. The source is connected to each of the intermediate nodes $V_S$ through a noiseless link of capacity $C_S$ (per channel use). On the other hand, each receiver $Y_i$ is connected to $2^{K-1}$ intermediate nodes $V_{\mathsf{W}_i^{\mathsf{P}}} \triangleq \{V_S\}_{S \in \mathsf{W}_i^{\mathsf{P}} }$ via noiseless links of unlimited capacity. An interesting connection between the combination networks and the DM BC is revealed in \cite{romero2016superposition} wherein the authors considered the combination network to be a network of noiseless DM BCs with the channel input $X$ connected in different ways to the channel outputs $\{Y_{i}\}_{i=1}^K$ each through a noiseless BC. In particular, the channel input $X$ contains $2^{K}{-}1$ components $V_S$, for all $S \in  \mathsf{P}$. For each $S$, the component $V_S \in \mathcal{V}_S$, where $|\mathcal{V}_S|=2^{C_S}$, is noiselessly received at each receiver $Y_i$ for all $i \in S$ and {\em not} received at the receivers $Y_j$ with $j\not\in S $, i.e., $Y_i= V_{\mathsf{W}_i^{\mathsf{P}}} $. 

Denote the set of the four messages $\{M_S : S \in \mathsf{E} \} $ to be sent over a $K$-user DM BC as $ M_{\mathsf{E}}$. A $ (\{2^{nR_{S}}\}_{S \in \mathsf{E}}, n) $ code consists of (i) an encoder that assigns to each message tuple $ m_{\mathsf{E}} \in \prod_{S \in \mathsf{E}}[1:2^{nR_S}] $ a codeword $x^n( m_\mathsf{E})$ (ii) a decoder at each receiver, with the $k^{th}$ decoder mapping the received sequence $Y_{k}^n$ for each $k \in [1:K]$ into the set of decoded messages $ \{ \tilde{m}_{S} : S \in \mathsf{W}_k^{\mathsf{E}} \} \in \prod_{S \in \mathsf{W}_k^{\mathsf{E}}}[1:2^{nR_S}]$, denoted as $ \tilde{m}_{\mathsf{W}_k^{\mathsf{E}}}$. The probability of error $P_e^{(n)}$ is the probability that not all receivers decode their intended messages correctly. The rate tuple $( R_{S} : S\in \mathsf{E} )$ is said to be achievable if there exists a sequence of $ (\{2^{nR_{S}}\}_{S \in \mathsf{E}},n) $ codes with $P_e^{(n)} \rightarrow 0$ as $n  \rightarrow  \infty $. The closure of the union of achievable rates is the capacity region.

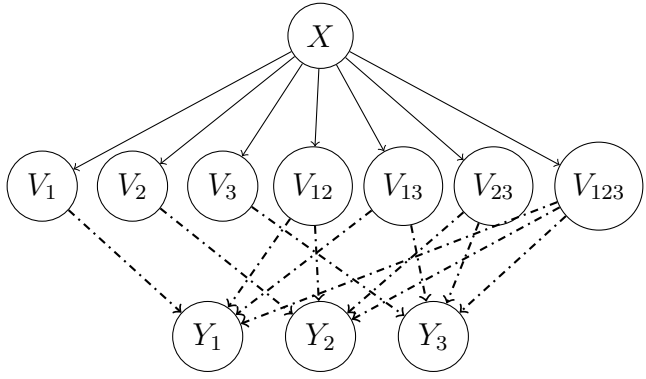
\begin{figure}
\centering
\begin{tikzpicture}



\node  (s) at (0,2.5)  [circle,draw] {$X$};
\node (v1) at (-3.7,0.5)   [circle,draw] {$V_{1}$};
\node (v2) at (-2.5,0.5)    [circle,draw] {$V_{2}$};
\node (v3) at (-1.3,0.5)    [circle,draw] {$V_{3}$};
\node (v12) at (-0.1,0.5)    [circle,draw] {$V_{12}$};
\node (v13) at (1.1,0.5)    [circle,draw] {$V_{13}$};
\node (v23) at (2.3,0.5)    [circle,draw] {$V_{23}$};
\node (v123) at (3.7,0.5)    [circle,draw] {$V_{123}$};

\node (y1) at (-1.5,-1.5)    [circle,draw] {$Y_1$};
\node (y2) at (0,-1.5)    [circle,draw] {$Y_2$};
\node (y3) at (1.5,-1.5)    [circle,draw] {$Y_3$};


\draw [->] (s) to  (v1);
\draw [->] (s) to  (v2);
\draw [->] (s) to  (v3);
\draw [->] (s) to  (v12);
\draw [->] (s) to  (v13);
\draw [->] (s) to  (v23);
\draw [->] (s) to  (v123);
\draw[thick,dash dot] [->] (v1) to  (y1);
\draw[thick,dash dot] [->] (v2) to  (y2);
\draw[thick,dash dot] [->] (v3) to  (y3);
\draw[thick,dash dot] [->] (v12) to  (y1);
\draw[thick,dash dot] [->] (v12) to  (y2);
\draw[thick,dash dot] [->] (v13) to  (y1);
\draw[thick,dash dot] [->] (v13) to  (y3);
\draw[thick,dash dot] [->] (v23) to  (y2);
\draw[thick,dash dot] [->] (v23) to  (y3);
\draw[thick,dash dot] [->] (v123) to  (y1);
\draw[thick,dash dot] [->] (v123) to  (y2);
\draw[thick,dash dot] [->] (v123) to  (y3);

\end{tikzpicture}
\caption{A combination network with $7$ intermediate nodes and three receivers. The dark/dashed lines represent finite/infinite capacity links, respectively. The capacity of the dark link connecting the node $X$ to the node $V_S$ is $C_S$ for each $S\in \mathsf{P}$. For brevity, the source/destination nodes are denoted by their transmitted/received symbols and the intermediate nodes by their output symbols. Encoders/decoders are not shown.
\label{Fig_K3_combination networks}
}
\end{figure}

\section{Results}
\label{Sec_Main_results}

In the following theorem, we present the inner bound of \cite[Theorem 1]{romero2016superposition} specialized to the diamond message set 
in the more explicit form of a union of four-dimensional 
polytopes.

\begin{theorem}
\label{Th_AchRegion_FourMsgs}
An inner bound of $K$-user DM BC for the diamond message set with $\mathsf{E}=\{\overline{\phi},\overline{K},\overline{K-1},\overline{K-1.K}  \}$ is the set of non-negative rate tuples ($R_{\overline{\phi}}, R_{\overline{K}},R_{\overline{K-1}},R_{\overline{K-1.K}} $) satisfying the following for all $j,j_1,j_2\in\{1,2,\cdots,K-2\}$
\begin{align}
&R_{\overline{\phi}}+ R_{\overline{K-1}}\leq I(U_{\overline{\phi}},U_{\overline{K-1}};Y_{K})
\label{Eq_Corollary_1}\\
&R_{\overline{\phi}}+ R_{\overline{K}}\leq I(U_{\overline{\phi}},U_{\overline{K}};Y_{K-1})
\label{Eq_Corollary_2}\\
&R_{\overline{\phi}}+ R_{\overline{K-1}}+ R_{\overline{K}} 
\leq 
I(U_{\overline{K}};Y_{K-1}|U_{\overline{\phi}}) \nonumber \\ 
& \hspace{3cm} +I(U_{\overline{\phi}},U_{\overline{K-1}};Y_{K})
\label{Eq_Corollary_3} 
\\
&R_{\overline{\phi}}+ R_{\overline{K-1}}+ R_{\overline{K}} 
\leq 
I(U_{\overline{K-1}};Y_{K}|U_{\overline{\phi}}) \nonumber \\ 
& \hspace{3cm} +I(U_{\overline{\phi}},U_{\overline{K}};Y_{K-1})
\label{Eq_Corollary_4} 
\\
&R_{\overline{\phi}}+ R_{\overline{K-1}}+ R_{\overline{K}} +R_{\overline{K-1.K}} 
\leq  
I(X;Y_j) 
\label{Eq_Corollary_5}
\\
&R_{\overline{\phi}}+ R_{\overline{K-1}}+ R_{\overline{K}} +R_{\overline{K-1.K}} 
\leq  
I(X;Y_j|U_{\overline{\phi}},U_{\overline{K-1}}) \nonumber \\ & \hspace{3cm} +I(U_{\overline{\phi}},U_{\overline{K-1}};Y_{K})
\label{Eq_Corollary_6} 
\\
&R_{\overline{\phi}}+ R_{\overline{K-1}}+ R_{\overline{K}} +R_{\overline{K-1.K}} 
\leq  
I(X;Y_j|U_{\overline{\phi}},U_{\overline{K}}) \nonumber \\ & \hspace{3cm} +I(U_{\overline{\phi}},U_{\overline{K}};Y_{K-1})
\label{Eq_Corollary_7} 
\\
&R_{\overline{\phi}}+ R_{\overline{K-1}}+ R_{\overline{K}} +R_{\overline{K-1.K}} 
\leq  
I(X;Y_j|U_{\overline{\phi}},U_{\overline{K-1}},U_{\overline{K}}) \nonumber \\ & \hspace{2.6cm} +I(U_{\overline{K}};Y_{K-1}|U_{\overline{\phi}}) +I(U_{\overline{\phi}},U_{\overline{K-1}};Y_{K})
\label{Eq_Corollary_8} 
\\
&R_{\overline{\phi}}+ R_{\overline{K-1}}+ R_{\overline{K}} +R_{\overline{K-1.K}} 
\leq  
I(X;Y_j|U_{\overline{\phi}},U_{\overline{K-1}},U_{\overline{K}}) \nonumber \\ & \hspace{2.6cm} +I(U_{\overline{K-1}};Y_{K}|U_{\overline{\phi}}) +I(U_{\overline{\phi}},U_{\overline{K}};Y_{K-1})
\label{Eq_Corollary_9} 
\\
&2R_{\overline{\phi}}+ R_{\overline{K-1}}+ R_{\overline{K}} +R_{\overline{K-1.K}} 
\leq  
I(X;Y_j|U_{\overline{\phi}},U_{\overline{K-1}},U_{\overline{K}}) \nonumber \\ & \hspace{2.6cm} +I(U_{\overline{\phi}},U_{\overline{K}};Y_{K-1})+I(U_{\overline{\phi}},U_{\overline{K-1}};Y_{K})
\label{Eq_Corollary_10} 
\\
&2R_{\overline{\phi}}+ 2R_{\overline{K-1}}+ 2R_{\overline{K}} +R_{\overline{K-1.K}} 
\leq  
I(X;Y_j|U_{\overline{\phi}}) \nonumber \\ & \hspace{2.6cm} +I(U_{\overline{\phi}},U_{\overline{K}};Y_{K-1})+I(U_{\overline{\phi}},U_{\overline{K-1}};Y_{K})
\label{Eq_Corollary_11} 
\\
&2R_{\overline{\phi}}+ 2R_{\overline{K-1}}+ 2R_{\overline{K}} +2R_{\overline{K-1.K}} 
\leq  
I(X;Y_{j_1}|U_{\overline{\phi}}) \nonumber \\ & \hspace{2.6cm}+   I(X;Y_{j_2}|U_{\overline{\phi}},U_{\overline{K-1}},U_{\overline{K}}) \nonumber \\ & \hspace{2.6cm} +I(U_{\overline{\phi}},U_{\overline{K}};Y_{K-1})+I(U_{\overline{\phi}},U_{\overline{K-1}};Y_{K})
\label{Eq_Corollary_12} 
\end{align}
for some joint distribution of the auxiliary and input random variables $(U_{\overline{\phi}}, U_{\overline{K}}, U_{\overline{K-1}},X)$ that is of the form
\begin{align}
    p(u_{\overline{\phi}}, u_{\overline{K}}, u_{\overline{K-1}},x) & = p(u_{\overline{\phi}}) \times p(u_{\overline{K}}|u_{\overline{\phi}}) \times p(u_{\overline{K-1}}|u_{\overline{\phi}}) \nonumber \\ & \quad \quad \quad \times p(x|u_{\overline{\phi}},u_{\overline{K}},u_{\overline{K-1}}) \label{pmfstructure}
    \end{align}
\end{theorem}
\begin{proof} 
We provide only an outline for the rate-splitting and superposition coding scheme of \cite[Theorem 1]{romero2016superposition}.
First, up-set rate splitting is used, i.e., 
each message $M_S$ is divided into a collection of sub-messages $ M_{S\rightarrow S'}$ where $S' \in \uparrow_{ \mathsf{E}} S$. Then, the sub-message $M_{S\rightarrow S^{'}}$ will be treated as if it was intended for the receivers in the larger set $S^{'}$ instead of in $S$. Hence, the messages $M_{\overline{K-1.K}},M_{\overline{K}}$ and $ M_{\overline{K-1}}$ are split into $(M_{\overline{K-1.K}\rightarrow\overline{\phi}  },M_{\overline{K-1.K}\rightarrow \overline{K}},M_{\overline{K-1.K}\rightarrow \overline{K-1}},M_{\overline{K-1.K} \rightarrow\overline{K-1.K}} )$, $(M_{\overline{K}\rightarrow \overline{\phi}}, M_{\overline{K}\rightarrow \overline{K}})$, $(M_{\overline{K-1} \rightarrow \overline{\phi}},M_{\overline{K-1} \rightarrow \overline{K-1}} )$, with rates $(R_{\overline{K-1.K}\rightarrow\overline{\phi}  },
R_{\overline{K-1.K}\rightarrow \overline{K}},
R_{\overline{K-1.K}\rightarrow \overline{K-1}},
R_{\overline{K-1.K} \rightarrow\overline{K-1.K}})$, $(R_{\overline{K}\rightarrow \overline{\phi}},
R_{\overline{K}\rightarrow \overline{K}})$, and 
$(R_{\overline{K-1} \rightarrow \overline{\phi}},
R_{\overline{K-1} \rightarrow \overline{K-1}} )$, respectively. The common message and a sub-message of each of the other three messages of the form $M_{S \rightarrow \bar{\phi}}$, that is, the reconstructed message $(M_{ \overline{\phi}},M_{\overline{K-1} \rightarrow \overline{\phi}}, M_{\overline{K} \rightarrow \overline{\phi}},M_{\overline{K-1.K} \rightarrow \overline{\phi}} )$ is represented by the cloud center $U_{\overline{\phi}}$. Conditionally independently (conditioned on $U_{\overline{\phi}}$), the reconstructed 
message pair $(M_{\overline{K-1}\rightarrow \overline{K-1}},M_{\overline{K-1.K}\rightarrow \overline{K-1}}  )$ is represented by $U_{\overline{K-1}}$ and $(M_{\overline{K}\rightarrow \overline{K}},M_{\overline{K-1.K} \rightarrow \overline{K}})$ is represented by $U_{\overline{K}}$. Finally, $M_{\overline{K-1.K}\rightarrow \overline{K-1.K}}$ is represented by $U_{\overline{K-1.K}}=X$. 
Receivers $\{Y_i\}_{i=1}^{K-2}$ decode the four intended messages by the joint unique decoding of $X$ (and hence $U_{\overline{\phi}},U_{\overline{K}},U_{\overline{K-1}},U_{\overline{K-1.K}}$),
and this happens successfully as long as
\begin{align}
    &R_{\overline{\phi}}+R_{\overline{K-1}}+R_{\overline{K}}+R_{\overline{K-1.K}} \leq I(X;Y_j) 
    \label{AchRegion_FME_0_j_1}\\
    &R_{\overline{K-1}}+R_{\overline{K}}+R_{\overline{K-1.K}} \nonumber \\
    & \quad -R_{\overline{K-1} \rightarrow \overline{\phi}}-R_{\overline{K}\rightarrow \overline{\phi}}-R_{\overline{K-1.K}\rightarrow \overline{\phi}} \leq I(X;Y_j|U_{\overline{\phi}})
    \label{AchRegion_FME_0_j_2}\\
    &R_{\overline{K}}+R_{\overline{K-1.K}}-R_{\overline{K} \rightarrow \overline{\phi}} \nonumber \\
    & \quad -R_{\overline{K-1.K}\rightarrow \overline{K-1}}-R_{\overline{K-1.K}\rightarrow \overline{\phi}} \leq I(X;Y_j|U_{\overline{\phi}},U_{\overline{K-1}})
    \label{AchRegion_FME_0_j_3}\\
    &R_{\overline{K-1}}+R_{\overline{K-1.K}}-R_{\overline{K-1} \rightarrow \overline{\phi}} \nonumber \\
    & \quad -R_{\overline{K-1.K}\rightarrow \overline{K}}-R_{\overline{K-1.K}\rightarrow \overline{\phi}} \leq I(X;Y_j|U_{\overline{\phi}},U_{\overline{K}})
    \label{AchRegion_FME_0_j_4}\\
    &R_{\overline{K-1.K}} -R_{\overline{K-1.K} \rightarrow \overline{K-1}}\nonumber \\
    & \quad -R_{\overline{K-1.K}\rightarrow \overline{K}}-R_{\overline{K-1.K}\rightarrow \overline{\phi}} \leq I(X;Y_j|U_{\overline{\phi}},U_{\overline{K-1}},U_{\overline{K}})
    \label{AchRegion_FME_0_j_5}
\end{align} for all $j\in \{1,2,\cdots,K-2\}$.


On the other hand, receiver $Y_{K-1}$ finds the two intended messages $M_{\overline{\phi}}$ and $ M_{\overline{K}}$ by decoding $(U_{\overline{\phi}},U_{\overline{K}})$, and receiver $Y_{K}$ find the two intended messages $M_{\overline{\phi}}$ and $M_{\overline{K-1}}$ by decoding  $(U_{\overline{\phi}},U_{\overline{K-1}})$. Receivers $Y_{K-1}$ and $Y_{K}$ decode their intended pair of messages successfully provided the following inequalities hold:
\begin{align}
    &R_{\overline{\phi}}+R_{\overline{K}}
     +R_{\overline{K-1} \rightarrow \overline{\phi} } \nonumber \\
    &+R_{\overline{K-1.K}\rightarrow \overline{K} }+R_{\overline{K-1.K} \rightarrow \overline{\phi} }\leq I(U_{\overline{\phi}},U_{\overline{K}};Y_{K-1})
    \label{AchRegion_FME_0_K-1_1}\\
    &R_{\overline{K}}-R_{\overline{K}\rightarrow \overline{\phi}} +R_{\overline{K-1.K}\rightarrow \overline{K}}\leq I(U_{\overline{K}};Y_{K-1}|U_{\overline{\phi}})
    \label{AchRegion_FME_0_K-1_2}\\
    &R_{\overline{\phi}}+R_{\overline{K-1}}
     +R_{\overline{K} \rightarrow \overline{\phi} } \nonumber \\
    &+R_{\overline{K-1.K}\rightarrow \overline{K-1} }+R_{\overline{K-1.K} \rightarrow \overline{\phi} }\leq I(U_{\overline{\phi}},U_{\overline{K-1}};Y_{K})
    \label{AchRegion_FME_0_K_1}\\
    &R_{\overline{K-1}}-R_{\overline{K-1}\rightarrow \overline{\phi}} +R_{\overline{K-1.K}\rightarrow \overline{K-1}}\leq I(U_{\overline{K-1}};Y_{K}|U_{\overline{\phi}})
    \label{AchRegion_FME_0_K_2}
\end{align}
Besides inequalities \eqref{AchRegion_FME_0_j_1}-\eqref{AchRegion_FME_0_K_2} there are eight inequalities to account for the non-negativity of the split rates given in \eqref{AchRegion_FME_0_positiveSplit_1}-\eqref{AchRegion_FME_0_positiveSplit_8} in Appendix \ref{Appendix_FME_TH1}. Taken together, these inequalities describe a nine-dimensional polytope (per admissible distribution of \eqref{pmfstructure}) in original- and split-rates space. By projecting away the five split rates 
using a structured form of the FME method described in Appendix \ref{Appendix_FME_TH1}, we get the achievable region in the positive orthant defined by the inequalities \eqref{Eq_Corollary_1}-\eqref{Eq_Corollary_12}. A high-level description of the issues encountered in performing the FME procedure is given in Remark \ref{rem:fme}. The four successively reduced-dimensional intermediate polytopes (of dimensions eight down to five) obtained after projecting out each split rate, in a certain desirable order, after the elimination of the associated redundant inequalities in this rather large FME procedure, are given in sequence in Appendix \ref{Appendix_FME_TH1}. That order permits a relatively efficient completion of the projection of all five split rates.
\end{proof}

\begin{remark}
\label{rem:fme}
Note that since $K$ is arbitrary, there are an indeterminate number of inequalities that describe the polyhedral achievable rate region per random coding distribution in nine-dimensional original- and split-rate space. In particular, including the non-negativity of split rates, there are $5K+2$ inequalities given by \eqref{AchRegion_FME_0_j_1}--\eqref{AchRegion_FME_0_K_2} and \eqref{AchRegion_FME_0_positiveSplit_1}-\eqref{AchRegion_FME_0_positiveSplit_8} that describe each such polytope. This means that the FME procedure must deal with groups of inequalities (for general $K$) at a time instead of exhaustively listing them 
(for this reason, the FME software of \cite{gattegno2016fourier} cannot be directly used).
More importantly, the elimination of the five split rates presents $5!=120$ possibilities for the orders in which to eliminate them, all of which must lead, in principle, to the right four-dimensional polytope described by the inequalities \eqref{Eq_Corollary_1}-\eqref{Eq_Corollary_12}. It turns out however, that at least some of these orders render the FME procedure too tedious to perform beyond even the second or third steps of the five-step process, including the identification (and hence, elimination) of the redundant inequalities that arise at each step. Appendix \ref{Appendix_FME_TH1} gives a ``good" order of elimination and the intermediate results obtained after projecting out each of the five split rates in that order that not only makes it possible to complete the five-step task with reasonable effort 
but also offers the possibility of verifying if two of the successively reduced-dimensional polytopes at the end of steps 2 and 4 satisfy certain symmetry conditions which, intuitively, they must (and which they do). 
\end{remark}



We next state the inner bound of Theorem \ref{Th_AchRegion_FourMsgs} when it is applied to the combination network by setting $X=V_{\mathsf{P}}$ and $Y_j = V_{\mathsf{W}_j^{\mathsf{P}}}$. The mutual information terms in the bounds are written as the difference between the entropy and conditional entropy so that we may get insights into optimizing the rate region over the coding distributions.

\begin{corollary}
\label{cor-comb}
An inner bound of $K$-user combination networks $\mathsf{E}=\{\overline{\phi},\overline{K},\overline{K-1},\overline{K-1.K}  \}$ is the set of non-negative rate tuples ($R_{\overline{\phi}}, R_{\overline{K}},R_{\overline{K-1}},R_{\overline{K-1.K}} $) satisfying the following for all $j,j_1,j_2\in\{1,2,\cdots,K-2\}$
\begin{align}
&R_{\overline{\phi}}+ R_{\overline{K-1}}
\leq 
H(V_{\mathsf{W}_K^{\mathsf{P}}}) -H(V_{\mathsf{W}_K^{\mathsf{P}}}|U_{\overline{\phi}},U_{\overline{K-1}})
\label{Eq_combination_network_inner_bound_H_1}\\
&R_{\overline{\phi}}+ R_{\overline{K}}
\leq 
H(V_{\mathsf{W}_{K-1}^{\mathsf{P}}}) -H(V_{\mathsf{W}_{K-1}^{\mathsf{P}}}|U_{\overline{\phi}},U_{\overline{K}})
\label{Eq_combination_network_inner_bound_H_2}\\
&R_{\overline{\phi}}+ R_{\overline{K-1}}+ R_{\overline{K}} 
\leq 
H(V_{\mathsf{W}_{K-1}^{\mathsf{P}}}|U_{\overline{\phi}}) -H(V_{\mathsf{W}_{K-1}^{\mathsf{P}}}|U_{\overline{\phi}},U_{\overline{K}})
 \nonumber \\ 
& \qquad \qquad \qquad \quad
+H(V_{\mathsf{W}_K^{\mathsf{P}}}) -H(V_{\mathsf{W}_K^{\mathsf{P}}}|U_{\overline{\phi}},U_{\overline{K-1}}) 
\label{Eq_combination_network_inner_bound_H_3} 
\\
&R_{\overline{\phi}}+ R_{\overline{K-1}}+ R_{\overline{K}} 
\leq 
H(V_{\mathsf{W}_K^{\mathsf{P}}}|U_{\overline{\phi}}) -H(V_{\mathsf{W}_K^{\mathsf{P}}}|U_{\overline{\phi}},U_{\overline{K-1}}) 
\nonumber \\ 
& \qquad \qquad \qquad \quad
+H(V_{\mathsf{W}_{K-1}^{\mathsf{P}}}) -H(V_{\mathsf{W}_{K-1}^{\mathsf{P}}}|U_{\overline{\phi}},U_{\overline{K}})
\label{Eq_combination_network_inner_bound_H_4} 
\\
&R_{\overline{\phi}}+ R_{\overline{K-1}}+ R_{\overline{K}} +R_{\overline{K-1.K}}
\leq  
H(V_{\mathsf{W}_j^{\mathsf{P}}})
\label{Eq_combination_network_inner_bound_H_5}
\\
&R_{\overline{\phi}}+ R_{\overline{K-1}}+ R_{\overline{K}} +R_{\overline{K-1.K}}
\leq  
H(V_{\mathsf{W}_j^{\mathsf{P}}}|U_{\overline{\phi}},U_{\overline{K-1}}) \nonumber \\ & \qquad \qquad \qquad \quad
+H(V_{\mathsf{W}_K^{\mathsf{P}}}) -H(V_{\mathsf{W}_K^{\mathsf{P}}}|U_{\overline{\phi}},U_{\overline{K-1}}) 
\label{Eq_combination_network_inner_bound_H_6} 
\\
&R_{\overline{\phi}}+ R_{\overline{K-1}}+ R_{\overline{K}} +R_{\overline{K-1.K}} 
\leq  
H(V_{\mathsf{W}_j^{\mathsf{P}}}|U_{\overline{\phi}},U_{\overline{K}}) \nonumber \\ & \qquad \qquad \qquad \quad
+H(V_{\mathsf{W}_{K-1}^{\mathsf{P}}}) -H(V_{\mathsf{W}_{K-1}^{\mathsf{P}}}|U_{\overline{\phi}},U_{\overline{K}})
\label{Eq_combination_network_inner_bound_H_7} 
\\
&R_{\overline{\phi}}+ R_{\overline{K-1}}+ R_{\overline{K}} +R_{\overline{K-1.K}}
\leq  
H(V_{\mathsf{W}_j^{\mathsf{P}}}|U_{\overline{\phi}},U_{\overline{K-1}},U_{\overline{K}}) 
\nonumber \\ & \qquad \qquad \qquad \quad
+H(V_{\mathsf{W}_{K-1}^{\mathsf{P}}}|U_{\overline{\phi}}) -H(V_{\mathsf{W}_{K-1}^{\mathsf{P}}}|U_{\overline{\phi}},U_{\overline{K}})
 \nonumber \\ 
& \qquad \qquad \qquad \quad
+H(V_{\mathsf{W}_K^{\mathsf{P}}}) -H(V_{\mathsf{W}_K^{\mathsf{P}}}|U_{\overline{\phi}},U_{\overline{K-1}})
\label{Eq_combination_network_inner_bound_H_8} 
\\
&R_{\overline{\phi}}+ R_{\overline{K-1}}+ R_{\overline{K}} +R_{\overline{K-1.K}}
\leq  
H(V_{\mathsf{W}_j^{\mathsf{P}}}|U_{\overline{\phi}},U_{\overline{K-1}},U_{\overline{K}}) 
\nonumber \\ & \qquad \qquad \qquad \quad
+H(V_{\mathsf{W}_K^{\mathsf{P}}}|U_{\overline{\phi}}) -H(V_{\mathsf{W}_K^{\mathsf{P}}}|U_{\overline{\phi}},U_{\overline{K-1}}) 
\nonumber \\ 
& \qquad \qquad \qquad \quad
+H(V_{\mathsf{W}_{K-1}^{\mathsf{P}}}) -H(V_{\mathsf{W}_{K-1}^{\mathsf{P}}}|U_{\overline{\phi}},U_{\overline{K}})
\label{Eq_combination_network_inner_bound_H_9} 
\\
&2R_{\overline{\phi}}+ R_{\overline{K-1}}+ R_{\overline{K}} +R_{\overline{K-1.K}} 
\leq  
H(V_{\mathsf{W}_j^{\mathsf{P}}}|U_{\overline{\phi}},U_{\overline{K-1}},U_{\overline{K}}) 
\nonumber \\ & \qquad \qquad \qquad \quad 
+H(V_{\mathsf{W}_K^{\mathsf{P}}}) -H(V_{\mathsf{W}_K^{\mathsf{P}}}|U_{\overline{\phi}},U_{\overline{K-1}}) 
\nonumber \\ & \qquad \qquad \qquad \quad 
+H(V_{\mathsf{W}_{K-1}^{\mathsf{P}}}) -H(V_{\mathsf{W}_{K-1}^{\mathsf{P}}}|U_{\overline{\phi}},U_{\overline{K}})
\label{Eq_combination_network_inner_bound_H_10} 
\\
&2R_{\overline{\phi}}+ 2R_{\overline{K-1}}+ 2R_{\overline{K}} +R_{\overline{K-1.K}} 
\leq  
H(V_{\mathsf{W}_j^{\mathsf{P}}}|U_{\overline{\phi}}) 
\nonumber \\ & \qquad \qquad \qquad \quad
+H(V_{\mathsf{W}_K^{\mathsf{P}}}) -H(V_{\mathsf{W}_K^{\mathsf{P}}}|U_{\overline{\phi}},U_{\overline{K-1}}) 
\nonumber \\ & \qquad \qquad \qquad \quad
+H(V_{\mathsf{W}_{K-1}^{\mathsf{P}}}) -H(V_{\mathsf{W}_{K-1}^{\mathsf{P}}}|U_{\overline{\phi}},U_{\overline{K}})
\label{Eq_combination_network_inner_bound_H_11} 
\\
&2R_{\overline{\phi}}+ 2R_{\overline{K-1}}+ 2R_{\overline{K}} +2R_{\overline{K-1.K}} 
\leq  
H(V_{\mathsf{W}_{j_1}^{\mathsf{P}}}|U_{\overline{\phi}})  
\nonumber \\ & \qquad \qquad \qquad \quad
+H(V_{\mathsf{W}_{j_2}^{\mathsf{P}}}|U_{\overline{\phi}},U_{\overline{K-1}},U_{\overline{K}}) 
\nonumber \\ & \qquad \qquad \qquad \quad
+H(V_{\mathsf{W}_K^{\mathsf{P}}}) -H(V_{\mathsf{W}_K^{\mathsf{P}}}|U_{\overline{\phi}},U_{\overline{K-1}}) 
\nonumber \\ & \qquad \qquad \qquad \quad
+H(V_{\mathsf{W}_{K-1}^{\mathsf{P}}}) -H(V_{\mathsf{W}_{K-1}^{\mathsf{P}}}|U_{\overline{\phi}},U_{\overline{K}})
\label{Eq_combination_network_inner_bound_H_12} 
\end{align}
for some joint distribution of the auxiliary and input random variables $(U_{\overline{\phi}}, U_{\overline{K}}, U_{\overline{K-1}},X=V_{\mathsf{P}})$ that is of the form
\begin{align}
    p(u_{\overline{\phi}}, u_{\overline{K}}, u_{\overline{K-1}},x) & = p(u_{\overline{\phi}}) \times p(u_{\overline{K}}|u_{\overline{\phi}}) \times p(u_{\overline{K-1}}|u_{\overline{\phi}}) \nonumber \\ & \quad \quad \quad \times p(x|u_{\overline{\phi}},u_{\overline{K}},u_{\overline{K-1}}) \label{pmfstructure_repeat}
    \end{align}
\end{corollary}

In the following theorem, we show that the inner bound in Theorem \ref{Th_AchRegion_FourMsgs} achieves the capacity region of the general combination network by identifying the optimal distribution for the auxiliary random variables $U_{\overline{\phi}},U_{\overline{K-1}},U_{\overline{K}}$ and the channel input components $\{V_{S}\}_{S\in \mathsf{P}}$. The optimality of that distribution is proved indirectly by showing that the resulting inequalities are particular members of the generic family of generalized cut-set (outer) bounds found in \cite{salimi2015generalized}.
\begin{theorem}
\label{Th_Capacity_combination_networks}
The capacity region of the $K$-user combination network for the diamond message set with $\mathsf{E}=\{M_{\overline{\phi}},M_{\overline{K}},M_{\overline{K-1}},M_{\overline{K-1.K}}  \}$ is the set of non-negative rate tuples ($R_{\overline{\phi}}, R_{\overline{K}},R_{\overline{K-1}},R_{\overline{K-1.K}} $) satisfying for all $j\in \{1,2,\cdots,K-2\}$
\begin{align}
&R_{\overline{\phi}}+ R_{\overline{K-1}}\leq C_{\mathsf{W}_{K}^{\mathsf{P}}}
\label{Eq_Capacity_CN_1}\\
&R_{\overline{\phi}}+ R_{\overline{K}}\leq C_{\mathsf{W}_{K-1}^{\mathsf{P}}}
\label{Eq_Capacity_CN_2}\\
&R_{\overline{\phi}}+ R_{\overline{K-1}}+ R_{\overline{K}} 
\leq 
C_{  \downarrow_{\mathsf{W}_{K-1}^{\mathsf{P}}} \{\overline{K}\} }
+ C_{\mathsf{W}_{K}^{\mathsf{P}}}
\label{Eq_Capacity_CN_3} 
\\
&R_{\overline{\phi}}+ R_{\overline{K-1}}+ R_{\overline{K}} +R_{\overline{K-1.K}} 
\leq  
C_{\mathsf{W}_j^{\mathsf{P}}}
\label{Eq_Capacity_CN_4}
\\
&2R_{\overline{\phi}}+ R_{\overline{K-1}}+ R_{\overline{K}} +R_{\overline{K-1.K}} 
\leq  \nonumber 
\\ & \qquad C_{  \downarrow_{\mathsf{W}_j^{\mathsf{P}}} \{\overline{K-1.K}\} } 
+ C_{\mathsf{W}_{K-1}^{\mathsf{P}}}
+ C_{\mathsf{W}_{K}^{\mathsf{P}}}
\label{Eq_Capacity_CN_5} 
\\
&2R_{\overline{\phi}}+ 2R_{\overline{K-1}}+ 2R_{\overline{K}} +R_{\overline{K-1.K}} 
\leq  \nonumber 
\\ & \qquad C_{  \downarrow_{\mathsf{W}_j^{\mathsf{P}}} \{\overline{K},\overline{K-1} \} } 
+ C_{\mathsf{W}_{K-1}^{\mathsf{P}}} 
+ C_{\mathsf{W}_{K}^{\mathsf{P}}}
\label{Eq_Capacity_CN_6} 
\end{align}  
where $C_{\mathsf{W}}=\sum_{S\in \mathsf{W}} C_S$ for any $ \mathsf{W} \subseteq \mathsf{P} $.
\end{theorem}

\begin{remark}
This capacity result generalizes three previously known results. First, it reproduces the capacity region for two nested (i.e., degraded) messages with at most two receivers demanding only the common messages $M_{\overline{\phi}}$ given in \cite[Theorems 4 and 5]{salman2018achievableAR} and \cite[Theorem 3]{bidokhti2016capacity} by setting $R_{\overline{K-1.K}}=R_{\overline{K-1}}=0$ and $R_{\overline{K}}=R_{\overline{K-1}}=0$ in \eqref{Eq_Capacity_CN_1}-\eqref{Eq_Capacity_CN_6}. Moreover, Theorem \ref{Th_Capacity_combination_networks} also recovers the capacity region for two messages each intended for $K-1$ receivers given in \cite[Theorem 3]{salman2018achievableAR} by setting $R_{\overline{\phi}}=R_{\overline{K-1.K}}=0$ in \eqref{Eq_Capacity_CN_1}-\eqref{Eq_Capacity_CN_6}. It must be noted here that the converse proofs in \cite{bidokhti2016capacity} and \cite{salman2018achievableAR} are established from ``first principles", and while they use the sub-modularity of entropy, they do not take a top-down approach as we do here in Section \ref{sec:converse} of making the connection to the generalized cut-set bound framework of \cite{salimi2015generalized}. Moreover, achievability for up to two common receivers was proved in \cite[Proposition 1]{bidokhti2016capacity} using linear network coding tailored to the combination network and two nested messages, whereas we take the top-down approach in Section \ref{sec:achievability-combnet} of  \cite{romero2016superposition,salman2018achievableAR} of starting from an inner bound for the DM BC and specializing it to the combination network.
\end{remark}

The special case of Theorem \ref{Th_Capacity_combination_networks} for three degraded messages (which is more general than \cite[Theorem 5]{salman2018achievableAR} and \cite[Theorem 3]{bidokhti2016capacity}) is given next. We simply set $R_{\overline{K-1}}=0$ in Theorem \ref{Th_Capacity_combination_networks} and note that \eqref{Eq_Capacity_CN_3} becomes redundant because of \eqref{Eq_Capacity_Three_Deg_CN_2} since $C_{  \downarrow_{\mathsf{W}_{K-1}^{\mathsf{P}}} \{\overline{K}\} }
+ C_{\mathsf{W}_{K}^{\mathsf{P}}}=C_{\mathsf{W}_{K-1}^{\mathsf{P}}\cup \mathsf{W}_{K}^{\mathsf{P}}}\geq C_{\mathsf{W}_{K-1}^{\mathsf{P}}}$

\begin{corollary}
The capacity region of the $K$-user combination network for the three degraded messages, i.e., $\mathsf{E}=\{M_{\overline{\phi}},M_{\overline{K}},M_{\overline{K-1.K}}  \}$ is the set of non-negative rate tuples ($R_{\overline{\phi}}, R_{\overline{K}},R_{\overline{K-1.K}} $) satisfying for all $j\in \{1,2,\cdots,K-2\}$
\begin{align}
&R_{\overline{\phi}} \leq C_{\mathsf{W}_{K}^{\mathsf{P}}}
\label{Eq_Capacity_Three_Deg_CN_1}\\
&R_{\overline{\phi}}+ R_{\overline{K}}\leq C_{\mathsf{W}_{K-1}^{\mathsf{P}}}
\label{Eq_Capacity_Three_Deg_CN_2}\\
&R_{\overline{\phi}}+ R_{\overline{K}} +R_{\overline{K-1.K}} 
\leq  
C_{\mathsf{W}_j^{\mathsf{P}}}
\label{Eq_Capacity_Three_Deg_CN_4}
\\
&2R_{\overline{\phi}}+  R_{\overline{K}} +R_{\overline{K-1.K}} 
\leq  \nonumber 
\\ & \qquad C_{  \downarrow_{\mathsf{W}_j^{\mathsf{P}}} \{\overline{K-1.K}\} } 
+ C_{\mathsf{W}_{K-1}^{\mathsf{P}}}
+ C_{\mathsf{W}_{K}^{\mathsf{P}}}
\label{Eq_Capacity_Three_Deg_CN_5} 
\\
&2R_{\overline{\phi}}+ 2R_{\overline{K}} +R_{\overline{K-1.K}} 
\leq  \nonumber 
\\ & \qquad C_{  \downarrow_{\mathsf{W}_j^{\mathsf{P}}} \{\overline{K},\overline{K-1} \} } 
+ C_{\mathsf{W}_{K-1}^{\mathsf{P}}} 
+ C_{\mathsf{W}_{K}^{\mathsf{P}}}
\label{Eq_Capacity_Three_Deg_CN_6} 
\end{align}  
\end{corollary}

\begin{remark}
The capacity region of the $K$-user combination network for the two degraded messages, i.e., $\mathsf{E}=\{M_{\overline{\phi}},M_{\overline{K-1.K}}  \}$ is the set of non-negative rate pairs ($R_{\overline{\phi}},R_{\overline{K-1.K}} $) satisfying for all $j\in \{1,2,\cdots,K-2\}$ the inequalities \eqref{Eq_Capacity_Three_Deg_CN_1}-\eqref{Eq_Capacity_Three_Deg_CN_5} (with $R_{\overline{K}}=0$). The inequality \eqref{Eq_Capacity_Three_Deg_CN_6} becomes redundant because of \eqref{Eq_Capacity_Three_Deg_CN_5} (with $R_{\overline{K}}=0$) since $C_{  \downarrow_{\mathsf{W}_j^{\mathsf{P}}} \{\overline{K-1.K}\} } \leq C_{  \downarrow_{\mathsf{W}_j^{\mathsf{P}}} \{\overline{K},\overline{K-1} \} } $. This capacity region coincides with the one found in \cite[Theorem 5]{salman2018achievableAR} and \cite[Theorem 3]{bidokhti2016capacity}.
\end{remark}

\subsection{Proof of Achievability for Theorem \ref{Th_Capacity_combination_networks}}
\label{sec:achievability-combnet}
One or both of the two negative conditional entropy terms, namely, $
H(V_{\mathsf{W}_K^{\mathsf{P}}}|U_{\overline{\phi}},U_{\overline{K-1}}) $ and $H(V_{\mathsf{W}_{K-1}^{\mathsf{P}}}|U_{\overline{\phi}},U_{\overline{K}})$ appear in all the bounds (except in \eqref{Eq_combination_network_inner_bound_H_5}) of the inequalities of Corollary \ref{cor-comb}. 
Hence those bounds can be upper bounded by setting $V_{\mathsf{W}_K^{\mathsf{P}}} = f (U_{\overline{\phi}},U_{\overline{K-1}})$ and $V_{\mathsf{W}_{K-1}^{\mathsf{P}}} = g (U_{\overline{\phi}},U_{\overline{K}})$ for some deterministic functions $f$ and $g$. 
Next, one or both of the positive entropy terms $ H(V_{\mathsf{W}_{K-1}^{\mathsf{P}}})$ and $ H(V_{\mathsf{W}_K^{\mathsf{P}}})$ appear in each of the bounds (except in \eqref{Eq_combination_network_inner_bound_H_5}) and all such terms can be maximized by taking the deterministic functions $f$ and $g$ to be identity maps so that $V_{\mathsf{W}_K^{\mathsf{P}}} = (U_{\overline{\phi}},U_{\overline{K-1}})$ and $V_{\mathsf{W}_{K-1}^{\mathsf{P}}} = (U_{\overline{\phi}},U_{\overline{K}})$ (thus this choice maximizes the first two bounds \eqref{Eq_combination_network_inner_bound_H_1} and \eqref{Eq_combination_network_inner_bound_H_2} of Corollary \ref{cor-comb}). 
Note that this choice also simultaneously maximizes the positive conditional entropy terms of the form $
H(V_{\mathsf{W}_j^{\mathsf{P}}}|U_{\overline{\phi}},U_{\overline{K-1}}) $ and $
H(V_{\mathsf{W}_j^{\mathsf{P}}}|U_{\overline{\phi}},U_{\overline{K}}) $ in the bounds of \eqref{Eq_combination_network_inner_bound_H_6} and \eqref{Eq_combination_network_inner_bound_H_7} and $
H(V_{\mathsf{W}_j^{\mathsf{P}}}|U_{\overline{\phi}},U_{\overline{K-1}}, U_{\overline{K}}) $ in the bounds of \eqref{Eq_combination_network_inner_bound_H_8}-\eqref{Eq_combination_network_inner_bound_H_11}. 
Next, consider the positive conditional entropy terms of the form $H(V_{\mathsf{W}_{K-1}^{\mathsf{P}}}|U_{\overline{\phi}}) $, 
$ H(V_{\mathsf{W}_K^{\mathsf{P}}}|U_{\overline{\phi}}) $ in the bounds of \eqref{Eq_combination_network_inner_bound_H_3}-\eqref{Eq_combination_network_inner_bound_H_4} and \eqref{Eq_combination_network_inner_bound_H_8}-\eqref{Eq_combination_network_inner_bound_H_9} and $H(V_{\mathsf{W}_{j}^{\mathsf{P}}}|U_{\overline{\phi}}) $ in \eqref{Eq_combination_network_inner_bound_H_11}-\eqref{Eq_combination_network_inner_bound_H_12}. 
All of these terms are maximized by choosing (given $V_{\mathsf{W}_K^{\mathsf{P}}} = (U_{\overline{\phi}},U_{\overline{K-1}})$ and $V_{\mathsf{W}_{K-1}^{\mathsf{P}}} = (U_{\overline{\phi}},U_{\overline{K}})$) $U_{\overline{\phi}}$ to be the just the common part of $V_{\mathsf{W}_{K-1}^{\mathsf{P}}}$ and $V_{\mathsf{W}_K^{\mathsf{P}}}$. After making these choices, given that there are cardinality constraints on the alphabet $\mathcal{V}_S$  of $V_S$ of $2^{C_S}$, all the bounds are maximized by choosing the input components $(V_S: S \in \mathsf{F})$ to be uniform and independent with $\mathcal{V}_S$ taken to be $[1:2^{C_S}]$.

The above argument suggests (even though it does not prove the optimality of) the choice of coding distribution resulting by setting 
\begin{align}
U_{\overline{\phi}} & =V_{\uparrow_{\mathsf{P}} \{K-1.K\} } \label{eq1optdist}
\\ U_{\overline{K-1}} &= V_{\mathsf{W}_{K}^{\mathsf{P}} } \label{eq2optdist} \\ U_{\overline{K}} &=V_{\mathsf{W}_{K-1}^{\mathsf{P}} } \label{eq3optdist} \\
X&=V_{\mathsf{P}} \label{eq4optdist}
\end{align}
and choosing the channel input components $V_S$ for all $S\in\mathsf{P}$ to be independent and uniformly distributed over $\mathcal{V}_S$ where $|\mathcal{V}_S|=2^{C_S}$. 
The achievability proof now follows from 
Theorem \ref{Th_AchRegion_FourMsgs} 
by choosing the above coding distribution. Note that such a choice of the auxiliary random variables and the channel input components
$ (U_{\overline{\phi}}, U_{\overline{K}}, U_{\overline{K-1}},X) $ has a joint distribution that satisfies \eqref{pmfstructure}, as it must. For this specific choice, we can compute the mutual information terms in \eqref{Eq_Corollary_1}-\eqref{Eq_Corollary_12} using $Y_i= V_{\mathsf{W}_i^{\mathsf{P}}} $ and \eqref{Eq_Lemma_3}-\eqref{Eq_Lemma_6} to obtain
\begin{align}
    I(U_{\overline{\phi}},U_{\overline{K-1}};Y_{K})&= C_{\mathsf{W}_K^{\mathsf{P}}}
    \label{Eq:compute_InnerBound_1}\\
    I(U_{\overline{\phi}},U_{\overline{K}};Y_{K-1})&= C_{\mathsf{W}_{K-1}^{\mathsf{P}}}
    \label{Eq:compute_InnerBound_2}\\
    I(U_{\overline{K}};Y_{K-1}|U_{\overline{\phi}})&= C_{ \downarrow_{\mathsf{W}_{K-1}^{\mathsf{P}}} \{\overline{K}\}}
    \label{Eq:compute_InnerBound_3}\\
    I(U_{\overline{K-1}};Y_{K}|U_{\overline{\phi}})&= C_{  \downarrow_{\mathsf{W}_{K}^{\mathsf{P}}} \{\overline{K-1}\} }
    \label{Eq:compute_InnerBound_4}\\
    I(X;Y_j) &= C_{\mathsf{W}_j^{\mathsf{P}}} 
    \label{Eq:compute_InnerBound_5}\\
    I(X;Y_j|U_{\overline{\phi}},U_{\overline{K}})&= C_{  \downarrow_{\mathsf{W}_j^{\mathsf{P}}} \{\overline{K}\} }
    \label{Eq:compute_InnerBound_6}\\
    I(X;Y_j|U_{\overline{\phi}},U_{\overline{K-1}})&= C_{  \downarrow_{\mathsf{W}_j^{\mathsf{P}}} \{\overline{K-1}\} }
    \label{Eq:compute_InnerBound_7}\\
    I(X;Y_j|U_{\overline{\phi}},U_{\overline{K-1}},U_{\overline{K}})&= C_{  \downarrow_{\mathsf{W}_j^{\mathsf{P}}} \{\overline{K-1.K}\} }
    \label{Eq:compute_InnerBound_8}\\
    I(X;Y_j|U_{\overline{\phi}})&= C_{  \downarrow_{\mathsf{W}_j^{\mathsf{P}}} \{\overline{K},\overline{K-1} \} }
    \label{Eq:compute_InnerBound_9}
\end{align} where $j\in\{1,2,\cdots, K-2\}$. Hence, we get \eqref{Eq_Capacity_CN_1}-\eqref{Eq_Capacity_CN_2} by substituting \eqref{Eq:compute_InnerBound_1}-\eqref{Eq:compute_InnerBound_2} in \eqref{Eq_Corollary_1}-\eqref{Eq_Corollary_2} and \eqref{Eq_Capacity_CN_3} by substituting \eqref{Eq:compute_InnerBound_3}-\eqref{Eq:compute_InnerBound_4} in \eqref{Eq_Corollary_3}-\eqref{Eq_Corollary_4} since $C_{  \downarrow_{\mathsf{W}_{K-1}^{\mathsf{P}}} \{\overline{K}\} }
+ C_{\mathsf{W}_{K}^{\mathsf{P}}}=C_{  \downarrow_{\mathsf{W}_{K}^{\mathsf{P}}} \{\overline{K-1}\} }
+ C_{\mathsf{W}_{K-1}^{\mathsf{P}}}$. Moreover, we get \eqref{Eq_Capacity_CN_4} directly by substituting \eqref{Eq:compute_InnerBound_5} in \eqref{Eq_Corollary_5}. Then, we obtain  \eqref{Eq_Capacity_CN_5} from \eqref{Eq_Corollary_10} using \eqref{Eq:compute_InnerBound_1}-\eqref{Eq:compute_InnerBound_2} and \eqref{Eq:compute_InnerBound_8}. Finally, we get \eqref{Eq_Capacity_CN_6} by substituting \eqref{Eq:compute_InnerBound_1}-\eqref{Eq:compute_InnerBound_2} and \eqref{Eq:compute_InnerBound_9} in \eqref{Eq_Corollary_11}. 

The rest of the inequalities in the inner bound, i.e., \eqref{Eq_Corollary_6}-\eqref{Eq_Corollary_9} and \eqref{Eq_Corollary_12} are redundant, as we show next. By computing the mutual information terms in \eqref{Eq_Corollary_6}-\eqref{Eq_Corollary_9} and \eqref{Eq_Corollary_12} for the specified coding distribution in \eqref{eq1optdist}-\eqref{eq4optdist}, it is left to the reader to verify that we obtain the following five categories of bounds (with each $j,j_1,j_2$ taking all possible values in the set $\{1,2,\cdots,K-2\}$) 
\begin{align}
&R_{\overline{\phi}}+ R_{\overline{K-1}}+ R_{\overline{K}} +R_{\overline{K-1.K}} 
\leq  
C_{  \downarrow_{\mathsf{W}_j^{\mathsf{P}}} \{\overline{K}\} } 
+ C_{\mathsf{W}_{K}^{\mathsf{P}}}
\label{Eq_Corollary_6_Redundant} 
\\
&R_{\overline{\phi}}+ R_{\overline{K-1}}+ R_{\overline{K}} +R_{\overline{K-1.K}} 
\leq 
C_{  \downarrow_{\mathsf{W}_j^{\mathsf{P}}} \{\overline{K-1}\} }
+ C_{\mathsf{W}_{K-1}^{\mathsf{P}}}
\label{Eq_Corollary_7_Redundant} 
\\
&R_{\overline{\phi}}+ R_{\overline{K-1}}+ R_{\overline{K}} +R_{\overline{K-1.K}} 
\leq  \nonumber \\ & \quad
C_{  \downarrow_{\mathsf{W}_j^{\mathsf{P}}} \{\overline{K-1.K}\} }+ 
C_{  \downarrow_{\mathsf{W}_{K-1}^{\mathsf{P}}} \{\overline{K}\} }
+ C_{\mathsf{W}_{K}^{\mathsf{P}}}
\label{Eq_Corollary_8_Redundant} 
\\
&R_{\overline{\phi}}+ R_{\overline{K-1}}+ R_{\overline{K}} +R_{\overline{K-1.K}} 
\leq \nonumber \\ & \quad
C_{  \downarrow_{\mathsf{W}_j^{\mathsf{P}}} \{\overline{K-1.K}\} }
+C_{  \downarrow_{\mathsf{W}_{K}^{\mathsf{P}}} \{\overline{K-1}\} } 
+ C_{\mathsf{W}_{K-1}^{\mathsf{P}}}
\label{Eq_Corollary_9_Redundant} 
\\
&2R_{\overline{\phi}}+ 2R_{\overline{K-1}}+ 2R_{\overline{K}} +2R_{\overline{K-1.K}} 
\leq \nonumber \\ & \quad
C_{  \downarrow_{\mathsf{W}_{j_1}^{\mathsf{P}}} \{\overline{K},\overline{K-1} \} }
+C_{  \downarrow_{\mathsf{W}_{j_1}^{\mathsf{P}}} \{\overline{K-1.K}\} }
+ C_{\mathsf{W}_{K-1}^{\mathsf{P}}}
+ C_{\mathsf{W}_{K}^{\mathsf{P}}}
\label{Eq_Corollary_12_Redundant} 
\end{align} 

We next show that these categories of bounds are all redundant. First, \eqref{Eq_Corollary_6_Redundant} and \eqref{Eq_Corollary_7_Redundant} are redundant from \eqref{Eq_Capacity_CN_4}. Using \eqref{Eq:Equality_Cs_1}, we can rewrite the right hand side of \eqref{Eq_Corollary_6_Redundant} and lower bound it as follows
\begin{align}
 C_{  \downarrow_{\mathsf{W}_j^{\mathsf{P}}} \{\overline{K}\} } 
+ C_{\mathsf{W}_{K}^{\mathsf{P}}}
&= 
C_{  \mathsf{W}_j^{\mathsf{P}} }
- C_{  \mathsf{W}_j^{\mathsf{P}} \cap  \mathsf{W}_{K}^{\mathsf{P}}}
+ C_{\mathsf{W}_{K}^{\mathsf{P}}}   \nonumber \\
&\geq C_{  \mathsf{W}_j^{\mathsf{P}} } 
\label{Eq:redundant_Show_1}
\end{align} where \eqref{Eq:redundant_Show_1} follows from the fact that $C_{  \mathsf{W}_j^{\mathsf{P}} \cap  \mathsf{W}_{K}^{\mathsf{P}}}
\leq C_{\mathsf{W}_{K}^{\mathsf{P}}}$. Similarly, we can show that \eqref{Eq_Corollary_7_Redundant} is redundant due to \eqref{Eq_Capacity_CN_4}. 

Moreover, \eqref{Eq_Corollary_8_Redundant} and \eqref{Eq_Corollary_9_Redundant} are also redundant from \eqref{Eq_Capacity_CN_4} since we can rewrite the right hand side of \eqref{Eq_Corollary_8_Redundant}, using \eqref{Eq:Equality_Cs_1}, and lower bound it as follows
\begin{align}
&C_{  \downarrow_{\mathsf{W}_j^{\mathsf{P}}} \{\overline{K-1.K}\} }
+ C_{  \downarrow_{\mathsf{W}_{K-1}^{\mathsf{P}}} \{\overline{K}\} } 
+ C_{\mathsf{W}_{K}^{\mathsf{P}}} \nonumber \\
&= 
C_{ \mathsf{W}_j^{\mathsf{P}}} -
C_{ \mathsf{W}_j^{\mathsf{P}}  \cap (\mathsf{W}_K^{\mathsf{P}} \cup \mathsf{W}_{K-1}^{\mathsf{P}}) }
+ C_{  \mathsf{W}_{K-1}^{\mathsf{P}}} 
- C_{  \mathsf{W}_{K-1}^{\mathsf{P}} \cap \mathsf{W}_{K-1}^{\mathsf{P}}} 
+ C_{\mathsf{W}_{K}^{\mathsf{P}}}\nonumber \\
&= 
C_{ \mathsf{W}_j^{\mathsf{P}}} -
C_{ \mathsf{W}_j^{\mathsf{P}}  \cap (\mathsf{W}_K^{\mathsf{P}} \cup \mathsf{W}_{K-1}^{\mathsf{P}}) }
+ C_{  \mathsf{W}_{K-1}^{\mathsf{P}} \cup \mathsf{W}_{K-1}^{\mathsf{P}}} 
\label{Eq:redundant_Show_2}\\
&\geq 
C_{ \mathsf{W}_j^{\mathsf{P}}} ,
\nonumber
\end{align} where \eqref{Eq:redundant_Show_2} follows from the modularity of $C_{\mathsf{W}}$. 

Finally, \eqref{Eq_Corollary_12_Redundant} is redundant from two times \eqref{Eq_Capacity_CN_4} since the right hand side of \eqref{Eq_Corollary_8_Redundant} can be written, using \eqref{Eq:Equality_Cs_1}-\eqref{Eq:Equality_Cs_2}, and lower bounded as follows
\begin{align}
&C_{  \downarrow_{\mathsf{W}_{j_1}^{\mathsf{P}}} \{\overline{K},\overline{K-1} \} }
+C_{  \downarrow_{\mathsf{W}_{j_2}^{\mathsf{P}}} \{\overline{K-1.K}\} }
+ C_{\mathsf{W}_{K-1}^{\mathsf{P}}}
+ C_{\mathsf{W}_{K}^{\mathsf{P}}}\nonumber \\
&= C_{  \mathsf{W}_{j_1}^{\mathsf{P}}} - C_{  \mathsf{W}_{j_1}^{\mathsf{P}} \cap \mathsf{W}_{K-1}^{\mathsf{P}}\cap \mathsf{W}_{K}^{\mathsf{P}}  } 
+C_{ \mathsf{W}_{j_2}^{\mathsf{P}}} - C_{ \mathsf{W}_{j_2}^{\mathsf{P}} \cap (\mathsf{W}_{K-1}^{\mathsf{P}}\cup \mathsf{W}_{K}^{\mathsf{P}}) } \nonumber \\ & \quad
+ C_{\mathsf{W}_{K-1}^{\mathsf{P}}}
+ C_{\mathsf{W}_{K}^{\mathsf{P}}}\nonumber \\
&= 
 C_{  \mathsf{W}_{j_1}^{\mathsf{P}}} 
-C_{  \mathsf{W}_{j_1}^{\mathsf{P}} \cap \mathsf{W}_{K-1}^{\mathsf{P}}\cap \mathsf{W}_{K}^{\mathsf{P}}  } 
+C_{  \mathsf{W}_{K-1}^{\mathsf{P}}\cap \mathsf{W}_{K}^{\mathsf{P}}  } \nonumber \\ & \quad
+C_{ \mathsf{W}_{j_2}^{\mathsf{P}}} 
-C_{ \mathsf{W}_{j_2}^{\mathsf{P}} \cap (\mathsf{W}_{K-1}^{\mathsf{P}}\cup \mathsf{W}_{K}^{\mathsf{P}}) } 
+C_{\mathsf{W}_{K-1}^{\mathsf{P}}\cup \mathsf{W}_{K}^{\mathsf{P}} } 
\label{Eq:redundant_Show_3}\\
&\geq 
 C_{  \mathsf{W}_{j_1}^{\mathsf{P}}} 
+C_{ \mathsf{W}_{j_2}^{\mathsf{P}}}, \nonumber 
\end{align}
where \eqref{Eq:redundant_Show_3} follows from the modularity of $C_{\mathsf{W}}$, and hence, $C_{\mathsf{W}_{K-1}^{\mathsf{P}}}
+ C_{\mathsf{W}_{K}^{\mathsf{P}}}= C_{  \mathsf{W}_{K-1}^{\mathsf{P}}\cap \mathsf{W}_{K}^{\mathsf{P}}  }+C_{\mathsf{W}_{K-1}^{\mathsf{P}}\cup \mathsf{W}_{K}^{\mathsf{P}} } $.

This concludes the achievability of the polytope described by the inequalities \eqref{Eq_Capacity_CN_1}-\eqref{Eq_Capacity_CN_6} by the single coding distribution given in \eqref{eq1optdist}-\eqref{eq4optdist}.



\begin{remark}
In Section \ref{sec:converse}, we show that the polytope (in the positive orthant) described by \eqref{Eq_Capacity_CN_1}-\eqref{Eq_Capacity_CN_6} is indeed the capacity region of the combination network. This indirectly shows that when the inner bound of Theorem \ref{Th_AchRegion_FourMsgs} is specialized to the combination network as a union of polytopes with the union taken over all admissible coding distributions of the form \eqref{pmfstructure} then the distribution given in \eqref{eq1optdist}-\eqref{eq4optdist} is extremal in that the rate region it produces subsumes the rate regions produced by any other distribution that satisfies \eqref{pmfstructure}.
\end{remark}


\subsection{The Converse Proof for Theorem \ref{Th_Capacity_combination_networks}}
\label{sec:converse}
The converse proof of \eqref{Eq_Capacity_CN_1}-\eqref{Eq_Capacity_CN_2} and \eqref{Eq_Capacity_CN_4} follows directly from the cut-set bounds. In what follows, we show that the inequalities \eqref{Eq_Capacity_CN_3} and \eqref{Eq_Capacity_CN_5}-\eqref{Eq_Capacity_CN_6} are outer bounds by identifying them as belonging to the infinite class of generalized cut-set bounds obtained in \cite{salimi2015generalized}. Those bounds provide a generic framework that can be used to obtain upper bounds on the achievable rates sent over any broadcast network that can be described by a directed acyclic graph with a non-negative capacity for each link or arc in the graph (and hence, over the combination network). As shown in \cite{salimi2015generalized}, these upper bounds take the generic form 
\begin{equation}
    \sum_{i\in S^{'}} \alpha_i R_{ \Phi_i ( \mathsf{W}_1^{\mathsf{E}},\mathsf{W}_2^{\mathsf{E}},\cdots, \mathsf{W}_K^{\mathsf{E}}  )} \leq \sum_{i\in S^{'}} \alpha_i C_{ \Phi_i ( \mathsf{W}_1^{\mathsf{P}},\mathsf{W}_2^{\mathsf{P}},\cdots, \mathsf{W}_K^{\mathsf{P}}  )}
    \label{Eq:GCSB_Upper_Bound}
\end{equation} where $S^{'}$ is a finite non-empty set, $\alpha_i$'s are non-negative real numbers and $\{\Phi_i\}_{i\in S^{'}} $ is a collection of set operators. The set operator is defined as a finite sequence of intersections and unions acting on any $K$ subsets of $\mathsf{P}$ to produce a subset of $\mathsf{P}$. 
Furthermore, these set operators $\{\Phi_i\}_{i\in S^{'} } $ in \eqref{Eq:GCSB_Upper_Bound} must satisfy the following extremal inequality
\begin{align}
    &\sum_{i\in S_1 }\alpha_i f( \Phi_i ( \mathsf{W}_1  ,\cdots, \mathsf{W}_K    )) 
    \leq
    \sum_{j\in S_2 }\beta_j f( \Pi_j ( \mathsf{W}_1   ,\cdots, \mathsf{W}_K    ))  \nonumber \\
    & \; 
    + \sum_{l\in S_3 }\gamma_l \big( f( \Gamma_l^+ ( \mathsf{W}_1    ,\cdots, \mathsf{W}_K    )) - f( \Gamma_l^- ( \mathsf{W}_1   ,\cdots, \mathsf{W}_K    ))  \big)
    \label{Eq:GCSB_General_Form}
\end{align} for any submodular function $f$, with equality holding when $f$ is a modular function, $\{\mathsf{W}_i \}_{i=1}^K\subseteq \mathsf{P}$, nonempty finite sets $S_1,S_2,S_3$, collection of non-negative real numbers ($\alpha_i,\beta_j,\gamma_l$), and a collection of set operators $\Pi_j, \Gamma_l^+,\Gamma_l^-$ such that $\Pi_j, \Gamma_l^+$ are finite sequences of unions only, and $\Gamma_l^- ( \mathsf{W}_1  ,\mathsf{W}_2  ,\cdots, \mathsf{W}_K    ) \subseteq \Gamma_l^+ ( \mathsf{W}_1  ,\mathsf{W}_2  ,\cdots, \mathsf{W}_K    )$ for any $l\in S_3$. 

The main problem in using this generic framework of infinitely many generalized cut-set bounds of the form \eqref{Eq:GCSB_Upper_Bound} is to identify the set operators $\Phi_i,\Pi_j,\Gamma_l^+,\Gamma_l^-$ for all $i\in S_1$, $ j\in S_2$ and $ l\in S_3$ that satisfy the extremal inequality \eqref{Eq:GCSB_General_Form} and provide specific upper bounds \eqref{Eq:GCSB_Upper_Bound} that exactly match with \eqref{Eq_Capacity_CN_3} and \eqref{Eq_Capacity_CN_5}-\eqref{Eq_Capacity_CN_6}, and thereby establishing the converse proof. 
With the aid of the order-theoretic relations \eqref{Eq_Lemma_1}-\eqref{Eq_Lemma_6}, we show that the following extremal inequalities proposed before in \cite[Proposition 1]{salimi2015generalized} provide tight upper bounds that match  \eqref{Eq_Capacity_CN_3} and \eqref{Eq_Capacity_CN_5}-\eqref{Eq_Capacity_CN_6}. From \cite[Proposition 1]{salimi2015generalized}, we get the following three extremal inequalities for $j\in \{1,2,\cdots, K-2\}$
\begin{align}
    &f(\mathsf{W}_{K-1}  \cup\mathsf{W}_{K}  )
    \leq
    f(\mathsf{W}_{K-1}  )+f(\mathsf{W}_{K}  )
    -f( \mathsf{W}_{K-1}  \cap\mathsf{W}_{K}  )
    \label{Eq:Extremal_inequality_0}\\
    &f(\mathsf{W}_{K-1}  \cup\mathsf{W}_{K}  \cup \mathsf{W}_{j}  )
    +f(\mathsf{W}_{K-1}  \cap\mathsf{W}_{K}    )
    \leq \nonumber \\
    &\; 
     f(\mathsf{W}_{j}  )
    -f(\mathsf{W}_{j}  \cap (\mathsf{W}_{K-1}  \cup\mathsf{W}_{K}  ))+f(\mathsf{W}_{K-1}  )+f(\mathsf{W}_{K}  )
    \label{Eq:Extremal_inequality_1}
    \\
    &f(\mathsf{W}_{K-1}  \cup\mathsf{W}_{K}  \cup \mathsf{W}_{j}  )
    +f( (\mathsf{W}_{K-1}  \cap\mathsf{W}_{K}  )\cup \nonumber \\
    & \quad (\mathsf{W}_{K-1}  \cap \mathsf{W}_{j}  ) \cup (\mathsf{W}_{K}  \cap \mathsf{W}_{j}  ) )
    \leq \nonumber \\
    &\; 
    f(\mathsf{W}_{j}  )
    -f(\mathsf{W}_{j}  \cap\mathsf{W}_{K-1}  \cap \mathsf{W}_{K}  )+f(\mathsf{W}_{K-1}  )+f(\mathsf{W}_{K}  )
    \label{Eq:Extremal_inequality_2}
    \end{align}

Using them, we get immediately from \eqref{Eq:GCSB_Upper_Bound} the following generalized cut-set bounds
\begin{align}
    &R_{\mathsf{W}_{K-1}^{\mathsf{E}}  \cup\mathsf{W}_{K}^{\mathsf{E}}  }
    \leq C_{\mathsf{W}_{K-1}^{\mathsf{P}}  \cup\mathsf{W}_{K}^{\mathsf{P}}  }
    \label{Eq:UpperBound_0}\\
    &R_{\mathsf{W}_{K-1}^{\mathsf{E}}  \cup\mathsf{W}_{K}^{\mathsf{E}}  \cup \mathsf{W}_{j}^{\mathsf{E}}  }
    +R_{\mathsf{W}_{K-1}^{\mathsf{E}}  \cap\mathsf{W}_{K}^{\mathsf{E}}    }
    \leq  \nonumber \\ & \quad
    C_{\mathsf{W}_{K-1}^{\mathsf{P}}  \cup\mathsf{W}_{K}^{\mathsf{P}}  \cup \mathsf{W}_{j}^{\mathsf{P}}  }
    +C_{\mathsf{W}_{K-1}^{\mathsf{P}}  \cap\mathsf{W}_{K}^{\mathsf{P}}    }
    \label{Eq:UpperBound_1}
    \\
    &R_{\mathsf{W}_{K-1}^{\mathsf{E}}  \cup\mathsf{W}_{K}^{\mathsf{E}}  \cup \mathsf{W}_{j}^{\mathsf{E}}  }
    +R_{ (\mathsf{W}_{K-1}^{\mathsf{E}}  \cap\mathsf{W}_{K}^{\mathsf{E}}  )\cup (\mathsf{W}_{K-1}^{\mathsf{E}}  \cap \mathsf{W}_{j}^{\mathsf{E}}  ) \cup (\mathsf{W}_{K}^{\mathsf{E}}  \cap \mathsf{W}_{j}^{\mathsf{E}}  ) }
    \leq \nonumber \\ & \quad
    C_{\mathsf{W}_{K-1}^{\mathsf{P}}  \cup\mathsf{W}_{K}^{\mathsf{P}}  \cup \mathsf{W}_{j}^{\mathsf{P}}  }
    +C_{ (\mathsf{W}_{K-1}^{\mathsf{P}}  \cap\mathsf{W}_{K}^{\mathsf{P}}  )\cup (\mathsf{W}_{K-1}^{\mathsf{P}}  \cap \mathsf{W}_{j}^{\mathsf{P}}  ) \cup (\mathsf{W}_{K}^{\mathsf{P}}  \cap \mathsf{W}_{j}^{\mathsf{P}}  ) }
    \label{Eq:upperBound_2}
\end{align}

By computing the left hand side of the above three inequalities, we get for any $j\in \{1,2,\cdots, K-2\}$
\begin{align}
    &R_{\mathsf{W}_{K-1}^{\mathsf{E}}  \cup\mathsf{W}_{K}^{\mathsf{E}}  }=R_{\overline{\phi}}+R_{\overline{K-1}}+R_{\overline{K}}
    \label{Eq:ComputeUpperBound_LHS_1}\\
    &R_{\mathsf{W}_{K-1}^{\mathsf{E}}  \cup\mathsf{W}_{K}^{\mathsf{E}}  \cup \mathsf{W}_{j}^{\mathsf{E}}  }
    +R_{\mathsf{W}_{K-1}^{\mathsf{E}}  \cap\mathsf{W}_{K}^{\mathsf{E}}    }
    =  \nonumber \\ & \quad
    2R_{\overline{\phi}}+R_{\overline{K-1}}+R_{\overline{K}}+R_{\overline{K-1.K}}
    \label{Eq:ComputeUpperBound_LHS_2}\\
    &R_{\mathsf{W}_{K-1}^{\mathsf{E}}  \cup\mathsf{W}_{K}^{\mathsf{E}}  \cup \mathsf{W}_{j}^{\mathsf{E}}  }
    +R_{ (\mathsf{W}_{K-1}^{\mathsf{E}}  \cap\mathsf{W}_{K}^{\mathsf{E}}  )\cup (\mathsf{W}_{K-1}^{\mathsf{E}}  \cap \mathsf{W}_{j}^{\mathsf{E}}  ) \cup (\mathsf{W}_{K}^{\mathsf{E}}  \cap \mathsf{W}_{j}^{\mathsf{E}}  ) }
    = \nonumber \\ & \quad
    2R_{\overline{\phi}}+2R_{\overline{K-1}}+2R_{\overline{K}}+R_{\overline{K-1.K}}
    \label{Eq:ComputeUpperBound_LHS_3}
\end{align}
These match with the left hand side of \eqref{Eq_Capacity_CN_3} and \eqref{Eq_Capacity_CN_5}-\eqref{Eq_Capacity_CN_6}.

Showing that the right hand sides of \eqref{Eq:UpperBound_0}-\eqref{Eq:upperBound_2} match with that of \eqref{Eq_Capacity_CN_3} and \eqref{Eq_Capacity_CN_5}-\eqref{Eq_Capacity_CN_6}, respectively, is less straightforward. We first use  \eqref{Eq_Lemma_1}-\eqref{Eq_Lemma_6} to rewrite the right hand side of \eqref{Eq_Capacity_CN_3} and \eqref{Eq_Capacity_CN_5}-\eqref{Eq_Capacity_CN_6} in a form that is closer to the generalized cut-set bounds. From \eqref{Eq_Lemma_1}-\eqref{Eq_Lemma_6}, we can show that the following two equalities hold for any set $S=l_1l_2\cdots l_N \subset \{1,2,\cdots,K\}$ and $j\in[1:K]$
\begin{align}
    C_{\downarrow_{\mathsf{W}_j^{\mathsf{P}}} \{\overline{S}\} }&= C_{\mathsf{W}_j^{\mathsf{P}}}- C_{\mathsf{W}_j^{\mathsf{P}} \cap (\mathsf{W}_{l_1}^{\mathsf{P}} \cup \mathsf{W}_{l_2}^{\mathsf{P}}\cup \cdots \cup \mathsf{W}_{l_N}^{\mathsf{P}}       ) } 
    \label{Eq:Equality_Cs_1}
    \\
    C_{\downarrow_{\mathsf{W}_j^{\mathsf{P}}} \{\overline{l_1},\overline{l_2},\cdots, \overline{l_N}\} }&= C_{\mathsf{W}_j^{\mathsf{P}}}- C_{\mathsf{W}_j^{\mathsf{P}} \cap (\mathsf{W}_{l_1}^{\mathsf{P}} \cap \mathsf{W}_{l_2}^{\mathsf{P}}\cap \cdots \cap \mathsf{W}_{l_N}^{\mathsf{P}}       ) } \label{Eq:Equality_Cs_2}
\end{align}
where $C_{\mathsf{W}}=\sum_{S\in \mathsf{W}} C_S$ for any $\mathsf{W}\subseteq \mathsf{P}$. Using \eqref{Eq:Equality_Cs_1}-\eqref{Eq:Equality_Cs_2}, we can rewrite the achievable bounds in \eqref{Eq_Capacity_CN_3} and \eqref{Eq_Capacity_CN_5}-\eqref{Eq_Capacity_CN_6} as follows
\begin{align}
&R_{\overline{\phi}}+ R_{\overline{K-1}}+ R_{\overline{K}} 
\leq 
 C_{\mathsf{W}_{K-1}^{\mathsf{P}}}+ C_{\mathsf{W}_{K}^{\mathsf{P}}}- C_{\mathsf{W}_{K-1}^{\mathsf{P}}\cap \mathsf{W}_{K}^{\mathsf{P}} }
\label{Eq_Capacity_CN_3_modified} 
\\
&2R_{\overline{\phi}}+ R_{\overline{K-1}}+ R_{\overline{K}} +R_{\overline{K-1.K}} 
\leq  \nonumber 
\\ & \qquad C_{\mathsf{W}_j^{\mathsf{P}}  }-C_{  \mathsf{W}_j^{\mathsf{P}}  \cap  (\mathsf{W}_{K-1}^{\mathsf{P}}  \cup \mathsf{W}_{K}^{\mathsf{P}} ) } 
+ C_{\mathsf{W}_{K-1}^{\mathsf{P}}}
+ C_{\mathsf{W}_{K}^{\mathsf{P}}}
\label{Eq_Capacity_CN_5_modified} 
\\
&2R_{\overline{\phi}}+ 2R_{\overline{K-1}}+ 2R_{\overline{K}} +R_{\overline{K-1.K}} 
\leq  \nonumber 
\\ & \qquad C_{\mathsf{W}_j^{\mathsf{P}}  }-C_{  \mathsf{W}_j^{\mathsf{P}}  \cap  (\mathsf{W}_{K-1}^{\mathsf{P}}  \cap \mathsf{W}_{K}^{\mathsf{P}} ) } 
+ C_{\mathsf{W}_{K-1}^{\mathsf{P}}} 
+ C_{\mathsf{W}_{K}^{\mathsf{P}}}
\label{Eq_Capacity_CN_6_modified} 
\end{align}  
By comparing \eqref{Eq:UpperBound_0}-\eqref{Eq:UpperBound_1} with \eqref{Eq_Capacity_CN_3_modified}-\eqref{Eq_Capacity_CN_6_modified}, the new representation of the achievable bound, we can see that the right hand side of both group of inequalities are identical using the extremal inequalities \eqref{Eq:Extremal_inequality_0}-\eqref{Eq:Extremal_inequality_2} for modular function since $C_{\mathsf{W}}$ is modular function for any $\mathsf{W}\subseteq \mathsf{P}$. Hence, the generalized cut-set bounds in \eqref{Eq:UpperBound_0}-\eqref{Eq:UpperBound_1} are tight. This concludes the converse proof.  


\begin{example}
\label{Example_K=4_CN_CapacityRegion}
The capacity region for the complete message set for the four-receiver combination network is not known. We specify the achievable rate region of Theorem \ref{Th_Capacity_combination_networks} for $K=4$ for the diamond message set, but without the structure in the order-theoretic specification of the bounds, as being the explicit polytope in the positive orthant described by the set of rate pairs ($R_{12},R_{123},R_{124},R_{1234}$) that satisfy the following nine linear inequalities 
\begin{align}
    &R_{1234}+R_{124}\leq C_4+C_{14}+C_{24}+C_{34}\nonumber \\ &\hspace{1.5cm} +C_{124}+C_{134}+C_{234}+C_{1234}
    \label{Eq:Example_K=4_Exhaustively_1}\\
    &R_{1234}+R_{123}\leq C_3+C_{13}+C_{23}+C_{34}\nonumber \\ &\hspace{1.5cm} +C_{123}+C_{134}+C_{234}+C_{1234}
    \label{Eq:Example_K=4_Exhaustively_2}\\
    &R_{1234}+R_{124}+R_{123}\leq C_3+C_4+C_{13} 
    \nonumber \\ &\hspace{1.5cm}
    +C_{14}+C_{23}+C_{24}+C_{34}+C_{123}
    \nonumber \\ &\hspace{1.5cm}
    +C_{124}+C_{134}+C_{234}+C_{1234}
    \label{Eq:Example_K=4_Exhaustively_3}\\
    &R_{1234}+R_{124}+R_{123}+R_{12}\leq 
    C_1+C_{12}+C_{13}+C_{14}\nonumber 
    \\ &\hspace{1.5cm} +C_{123}+C_{124}+C_{134}+C_{1234}
    \label{Eq:Example_K=4_Exhaustively_4}\\
    &R_{1234}+R_{124}+R_{123}+R_{12}\leq 
    C_2+C_{12}+C_{23}+C_{24}\nonumber \\
    &\hspace{1.5cm} +C_{123}+C_{124}+C_{234}+C_{1234}
    \label{Eq:Example_K=4_Exhaustively_5}\\
    &2R_{1234}+R_{124}+R_{123}+R_{12}\leq C_1+C_3+C_4+C_{12}
    \nonumber \\ &\hspace{1.5cm}
    +C_{13} +C_{14}+C_{23}+C_{24}+2C_{34}+C_{123}
    \nonumber \\ &\hspace{1.5cm}
    +C_{124}+2C_{134}+2C_{234}+2C_{1234}
    \label{Eq:Example_K=4_Exhaustively_6}\\
    &2R_{1234}+R_{124}+R_{123}+R_{12}\leq C_2+C_3+C_4+C_{12}
    \nonumber \\ &\hspace{1.5cm}
    +C_{13} +C_{14}+C_{23}+C_{24}+2C_{34}+C_{123}
    \nonumber \\ &\hspace{1.5cm}
    +C_{124}+2C_{134}+2C_{234}+2C_{1234}
    \label{Eq:Example_K=4_Exhaustively_7}\\
    &2R_{1234}+2R_{124}+2R_{123}+R_{12}\leq C_1+C_3+C_4+C_{12}
    \nonumber \\ &\hspace{1.5cm}
    +2C_{13} +2C_{14}+C_{23}+C_{24}+2C_{34}+2C_{123}
    \nonumber \\ &\hspace{1.5cm}
    +2C_{124}+2C_{134}+2C_{234}+2C_{1234}
    \label{Eq:Example_K=4_Exhaustively_8}\\
    &2R_{1234}+2R_{124}+2R_{123}+R_{12}\leq C_2+C_3+C_4+C_{12}
    \nonumber \\ &\hspace{1.5cm}
    +C_{13} +C_{14}+2C_{23}+2C_{24}+2C_{34}+2C_{123}
    \nonumber \\ &\hspace{1.5cm}
    +2C_{124}+2C_{134}+2C_{234}+2C_{1234}
    \label{Eq:Example_K=4_Exhaustively_9}
\end{align}

\end{example}

\section{Inner Bound For the DM BC With Binning}
\label{sec:binning}
For the sake of completeness and further investigation, we generalize Theorem \ref{Th_AchRegion_FourMsgs} by enhancing the achievable scheme therein by adding to it the technique of binning.

\begin{theorem}
\label{thm:wbinning}
An inner bound of $K$-user DM BC for the diamond message set with $\mathsf{E}=\{\overline{\phi},\overline{K},\overline{K-1},\overline{K-1.K}  \}$ is the set of non-negative rate tuples ($R_{\overline{\phi}}, R_{\overline{K}},R_{\overline{K-1}},R_{\overline{K-1.K}} $) satisfying the following for all $j,j_1,j_2\in\{1,2,\cdots,K-2\}$
\begin{align}
&R_{\overline{\phi}}+ R_{\overline{K-1}}
\leq 
I(U_{\overline{\phi}},U_{\overline{K-1}};Y_{K})
\label{Eq_Th_binning_1}\\
&R_{\overline{\phi}}+ R_{\overline{K}}
\leq 
I(U_{\overline{\phi}},U_{\overline{K}};Y_{K-1})
\label{Eq_Th_binning_2}\\
&R_{\overline{\phi}}+ R_{\overline{K-1}}+ R_{\overline{K}} 
\leq 
I(U_{\overline{K}};Y_{K-1}|U_{\overline{\phi}}) 
+I(U_{\overline{\phi}},U_{\overline{K-1}};Y_{K})
\nonumber \\  & \hspace{3cm} 
-I(U_{\overline{K}};U_{\overline{K-1}}|U_{\overline{\phi}}) 
\label{Eq_Th_binning_3} 
\\
&R_{\overline{\phi}}+ R_{\overline{K-1}}+ R_{\overline{K}} 
\leq 
I(U_{\overline{K-1}};Y_{K}|U_{\overline{\phi}})
+I(U_{\overline{\phi}},U_{\overline{K}};Y_{K-1})
\nonumber \\  & \hspace{3cm} 
-I(U_{\overline{K}};U_{\overline{K-1}}|U_{\overline{\phi}})
\label{Eq_Th_binning_4} 
\\
&2R_{\overline{\phi}}+ R_{\overline{K-1}}+ R_{\overline{K}}
\leq 
I(U_{\overline{\phi}},U_{\overline{K}};Y_{K-1})
\nonumber \\  & \hspace{1.5cm} 
+I(U_{\overline{\phi}},U_{\overline{K-1}};Y_{K})
-I(U_{\overline{K}};U_{\overline{K-1}}|U_{\overline{\phi}}) 
\label{Eq_Th_binning_5}
\\
&R_{\overline{\phi}}+ R_{\overline{K-1}}+ R_{\overline{K}} +R_{\overline{K-1.K}} 
\leq  
I(X;Y_j) 
\label{Eq_Th_binning_6}
\\
&R_{\overline{\phi}}+ R_{\overline{K-1}}+ R_{\overline{K}} +R_{\overline{K-1.K}} 
\leq  
I(X;Y_j|U_{\overline{\phi}},U_{\overline{K-1}}) \nonumber \\ & \hspace{4.5cm} +I(U_{\overline{\phi}},U_{\overline{K-1}};Y_{K})
\label{Eq_Th_binning_7} 
\\
&R_{\overline{\phi}}+ R_{\overline{K-1}}+ R_{\overline{K}} +R_{\overline{K-1.K}} 
\leq  
I(X;Y_j|U_{\overline{\phi}},U_{\overline{K}}) \nonumber \\ & \hspace{4.5cm} +I(U_{\overline{\phi}},U_{\overline{K}};Y_{K-1})
\label{Eq_Th_binning_8} 
\\
&R_{\overline{\phi}}+ R_{\overline{K-1}}+ R_{\overline{K}} +R_{\overline{K-1.K}} 
\leq  
I(X;Y_j|U_{\overline{\phi}},U_{\overline{K-1}},U_{\overline{K}}) 
\nonumber \\ & \hspace{2.6cm} 
+I(U_{\overline{K}};Y_{K-1}|U_{\overline{\phi}}) +I(U_{\overline{\phi}},U_{\overline{K-1}};Y_{K})
\nonumber \\ & \hspace{4.5cm} 
-I(U_{\overline{K}};U_{\overline{K-1}}|U_{\overline{\phi}}) 
\label{Eq_Th_binning_9} 
\\
&R_{\overline{\phi}}+ R_{\overline{K-1}}+ R_{\overline{K}} +R_{\overline{K-1.K}} 
\leq  
I(X;Y_j|U_{\overline{\phi}},U_{\overline{K-1}},U_{\overline{K}})
\nonumber \\ & \hspace{2.5cm} 
+I(U_{\overline{K-1}};Y_{K}|U_{\overline{\phi}}) +I(U_{\overline{\phi}},U_{\overline{K}};Y_{K-1})
\nonumber \\ & \hspace{4.5cm} 
-I(U_{\overline{K}};U_{\overline{K-1}}|U_{\overline{\phi}}) 
\label{Eq_Th_binning_10} 
\\
&2R_{\overline{\phi}}+ R_{\overline{K-1}}+ R_{\overline{K}} +R_{\overline{K-1.K}} 
\leq  
I(X;Y_j|U_{\overline{\phi}},U_{\overline{K-1}},U_{\overline{K}})
\nonumber \\ & \hspace{2.6cm} 
+I(U_{\overline{\phi}},U_{\overline{K}};Y_{K-1})
+I(U_{\overline{\phi}},U_{\overline{K-1}};Y_{K})
\nonumber \\ & \hspace{4.5cm} 
-I(U_{\overline{K}};U_{\overline{K-1}}|U_{\overline{\phi}}) 
\label{Eq_Th_binning_11} 
\\
&2R_{\overline{\phi}}+ 2R_{\overline{K-1}}+ 2R_{\overline{K}} +R_{\overline{K-1.K}} 
\leq  
I(X;Y_j|U_{\overline{\phi}}) 
\nonumber \\ & \hspace{2.6cm} 
+I(U_{\overline{\phi}},U_{\overline{K}};Y_{K-1})
+I(U_{\overline{\phi}},U_{\overline{K-1}};Y_{K})
\nonumber \\ & \hspace{4.5cm} 
-I(U_{\overline{K}};U_{\overline{K-1}}|U_{\overline{\phi}}) 
\label{Eq_Th_binning_12} 
\\
&2R_{\overline{\phi}}+ 2R_{\overline{K-1}}+ 2R_{\overline{K}} +2R_{\overline{K-1.K}} 
\leq  
I(X;Y_{j_1}|U_{\overline{\phi}}) 
\nonumber \\ & \hspace{2.6cm}+   I(X;Y_{j_2}|U_{\overline{\phi}},U_{\overline{K-1}},U_{\overline{K}}) 
\nonumber \\ & \hspace{2.6cm} 
+I(U_{\overline{\phi}},U_{\overline{K}};Y_{K-1})
+I(U_{\overline{\phi}},U_{\overline{K-1}};Y_{K})
\nonumber \\ & \hspace{4.5cm} 
-I(U_{\overline{K}};U_{\overline{K-1}}|U_{\overline{\phi}}) 
\label{Eq_Th_binning_13}\\
&2R_{\overline{\phi}}+ R_{\overline{K-1}}+ 2R_{\overline{K}} +R_{\overline{K-1.K}} 
\leq  
I(X;Y_j|U_{\overline{\phi}},U_{\overline{K-1}}) 
\nonumber \\ & \hspace{2.6cm}
+I(U_{\overline{\phi}},U_{\overline{K-1}};Y_{K})
+I(U_{\overline{\phi}},U_{\overline{K}};Y_{K-1})
\nonumber \\ & \hspace{4.5cm} 
-I(U_{\overline{K}};U_{\overline{K-1}}|U_{\overline{\phi}}) 
\label{Eq_Th_binning_14} 
\\
&2R_{\overline{\phi}}+ 2R_{\overline{K-1}}+ R_{\overline{K}} + R_{\overline{K-1.K}} 
\leq  
I(X;Y_j|U_{\overline{\phi}},U_{\overline{K}})
\nonumber \\ & \hspace{2.6cm}
+I(U_{\overline{\phi}},U_{\overline{K-1}};Y_{K})
+I(U_{\overline{\phi}},U_{\overline{K}};Y_{K-1})
\nonumber \\ & \hspace{4.5cm} 
-I(U_{\overline{K}};U_{\overline{K-1}}|U_{\overline{\phi}}) 
\label{Eq_Th_binning_15} 
\\
&2R_{\overline{\phi}}+ 2R_{\overline{K-1}}+ 2R_{\overline{K}} +2 R_{\overline{K-1.K}} 
\leq  
I(X;Y_{j_1}|U_{\overline{\phi}},U_{\overline{K}})
\nonumber \\ & \hspace{1.5cm}
+I(X;Y_{j_2}|U_{\overline{\phi}},U_{\overline{K-1}}) +I(U_{\overline{\phi}},U_{\overline{K-1}};Y_{K})
\nonumber \\ & \hspace{1.5cm}
+I(U_{\overline{\phi}},U_{\overline{K}};Y_{K-1}) -I(U_{\overline{K}};U_{\overline{K-1}}|U_{\overline{\phi}}) 
\label{Eq_Th_binning_16} 
\end{align}
for some joint distribution of the auxiliary and input random variables $(U_{\overline{\phi}}, U_{\overline{K}}, U_{\overline{K-1}},X)$ that is of the form
\begin{align}
    p(u_{\overline{\phi}}, u_{\overline{K}}, u_{\overline{K-1}},x) = p(u_{\overline{\phi}}) & \times p(u_{\overline{K}},u_{\overline{K-1}}|u_{\overline{\phi}}) \nonumber \\ & \times p(x|u_{\overline{\phi}},u_{\overline{K}},u_{\overline{K-1}}) 
    \label{pmfstructure_binning}
    \end{align}
\end{theorem}
\begin{proof}
An outline of the proof is given in Appendix \ref{app:thmbinning}.
\end{proof}

\begin{remark}
By setting $I(U_{\overline{K}};U_{\overline{K-1}}|U_{\overline{\phi}}) =0$ (i.e., no binning), we get that \eqref{Eq_Th_binning_5} is redundant from \eqref{Eq_Th_binning_1} and \eqref{Eq_Th_binning_2}, \eqref{Eq_Th_binning_14} is redundant from \eqref{Eq_Th_binning_7} and \eqref{Eq_Th_binning_2}, \eqref{Eq_Th_binning_15} is redundant from \eqref{Eq_Th_binning_8} and \eqref{Eq_Th_binning_1}, and \eqref{Eq_Th_binning_16} is redundant from \eqref{Eq_Th_binning_7} and \eqref{Eq_Th_binning_8}, thereby recovering Theorem \ref{Th_AchRegion_FourMsgs}. While the greater generality offered by binning is not necessary in establishing the capacity region of the combination network for the diamond message set it is an interesting open question as to whether the above achievable rate region is tight for say the more general deterministic DM BC where each output $Y_i$ is some deterministic function of $g_i(X)$.
\end{remark}

\begin{remark}
In the achievable rate region for the $K=2$ case of Theorem \ref{thm:wbinning} obtained for the complete message set $\mathsf{E}=\{1,2,12\}$ by setting $R_{\overline{K-1.K}}=0$ (in this case $X$ is a function of $U_1,U_2$ and  $U_{12}$) and $Y_j={\rm const.}$ for $j \in [3:K]$ it is easily seen that all but the first five inequalities \eqref{Eq_Th_binning_1}-\eqref{Eq_Th_binning_5} are redundant. In particular, this special case of Theorem \ref{thm:wbinning} recovers \cite[Theorem 5]{liang2007rate}, where it was first obtained, as it must, and is also given in \cite[Theorem 8.4]{el2011network}. 
\end{remark}

\begin{remark} The special case of $K=3$ and three degraded messages with $\mathsf{E}= \{1,12,123\} $ which is a subset of the diamond message set $\mathsf{E}= \{1,12,13, 123\} $ can be deduced by setting $K=3$ and $R_{13}=0$ in Theorem \ref{thm:wbinning}. This special case was also addressed in \cite[Theorem 2]{nair2009capacity}. There are some notable differences between the two regions however\footnote{There is a typographical error in the last inequality in the rate region given in \cite[Theorem 2]{nair2009capacity}: the term $I(U_{12};Y_2|U_{123})$ in the bound should instead be $I(U_{12},U_{123};Y_2)$.}. First, the distributions for the input and auxiliary random variables is taken to belong to a more restrictive class in \cite[Theorem 2]{nair2009capacity} that can be obtained by setting $U_{123}=U$, $(U_{123}, U_{12})=V_2$ and $(U_{123}, U_{13})=V_3$ and restricting to joint distributions wherein we have the Markov chains $ U \markov V_2 \markov (V_3, X)$ and  $ U \markov V_3 \markov (V_2, X)$ which is a stronger restriction than the one herein which is that  $ U_{123} \markov (U_{12}, U_{13})  \markov X$. However, the restriction in the class of coding distributions in \cite[Theorem 2]{nair2009capacity} is without loss of generality as argued in \cite{nair2009capacity}. More importantly however, the achievable region of \cite[Theorem 2]{nair2009capacity} corresponds to one wherein Receiver 3 decodes the common message $M_{123}$ non-uniquely via $V_3$ whereas in the specialization of Theorem \ref{thm:wbinning} Receiver 3 decodes the common message uniquely via $U_{123}$. This results in more inequalities in the region of Theorem \ref{thm:wbinning} (when specialized to $K=3$ and $\mathsf{E}= \{1,12,123\} $) than the ones in \cite[Theorem 2]{nair2009capacity}. However, it is left to the reader to verify along the lines of \cite[Proposition 8]{nair2009capacity} that unique and non-unique decoding at Receiver 3 result in the same region per coding distribution. In effect, the additional inequalities due to unique-decoding are redundant.
In particular, \cite[Proposition 8]{nair2009capacity} shows that the achievable region of \cite[Theorem 2]{nair2009capacity} for the three degraded message set in which Receiver 2 decodes the common message uniquely and Receiver 3 decodes the common message non-uniquely, when specialized to the two degraded message set $\mathsf{E}= \{1, 123\} $ (by setting $R_{12}=0$), is equivalent to the achievable rate region of \cite[Proposition 5]{nair2009capacity} obtained by having both Receivers 2 and 3 decode the common message non-uniqely.
\end{remark}

\section{Conclusion}
\label{Sec_Conc}
In this paper, we propose an inner bound based on rate-splitting and superposition coding in explicit form for the $K$-user DM BC as a union of four-dimensional polytopes described in terms of the rates of the four messages. 
By specializing it to the so-called combination network through the choice of a single yet optimal coding distribution we obtain a single polytope which is then shown to be the capacity region of that network for the diamond message set. Our converse proof uses the generalized cut-set bound framework of \cite{salimi2015generalized} to recognize, by developing key connections between our order-theoretic description of the inner bound and the 
set-theoretic description of the generalized cut-set bounds, the inequalities in the inner bound as indeed being generalized cut-set outer bounds. Most significantly, our interpretation of this capacity result is that it suggests a certain strength of the coding scheme for the {\em general} DM BC itself, from whence it was obtained. Put differently, our top-down approach allows us to conclude that rate-splitting and superposition coding (with subset inclusion order) results in a strong inner bound for the DM BC with the diamond message set as opposed to presenting the more conventional view that rate-splitting and superposition coding may not be strong enough to attain the capacity of the entire DM BC for the diamond message set. Nevertheless, we extend the rate-splitting and superposition coding inner bound by adding binning to it and obtain a rate region that is again given as a union of explicit four-dimensional polytopes in terms of the four message rates.

Moreover, we expect that the proposed top-town approach, initiated in \cite{romero2016superposition} and continued in \cite{salman2018achievableAR} and the present work, can be used to obtain inner bounds for the DM BC for message sets other than the diamond message set, and validate them as being strong, provided that, when specialized to the combination network, they establish its capacity region. 


\begin{appendices}
\section{The FME for Theorem \ref{Th_AchRegion_FourMsgs}}
\label{Appendix_FME_TH1}


Besides the inequalities in \eqref{AchRegion_FME_0_j_1}-\eqref{AchRegion_FME_0_K_2}, 
we have the following inequalities to ensure the non-negativity of the split rates
\begin{align}
    & -R_{\overline{K-1.K}}+ \nonumber \\
    & \qquad R_{\overline{K-1.K}\rightarrow\overline{K-1}}+R_{\overline{K-1.K}\rightarrow\overline{K}}+R_{\overline{K-1.K}\rightarrow\overline{\phi}}\leq 0
    \label{AchRegion_FME_0_positiveSplit_1}\\
    &-R_{\overline{K-1.K}\rightarrow\overline{K-1}} \leq 0
    \label{AchRegion_FME_0_positiveSplit_2}\\
    &-R_{\overline{K-1.K}\rightarrow\overline{K}} \leq 0
    \label{AchRegion_FME_0_positiveSplit_3}\\
    &-R_{\overline{K-1.K}\rightarrow\overline{\phi}} \leq 0
    \label{AchRegion_FME_0_positiveSplit_4}\\
    & -R_{\overline{K}}+R_{\overline{K}\rightarrow\overline{\phi}}\leq 0
    \label{AchRegion_FME_0_positiveSplit_5}\\
    &-R_{\overline{K}\rightarrow\overline{\phi}} \leq 0
    \label{AchRegion_FME_0_positiveSplit_6}\\
    & -R_{\overline{K-1}}+R_{\overline{K-1}\rightarrow\overline{\phi}}\leq 0
    \label{AchRegion_FME_0_positiveSplit_7}\\
    &-R_{\overline{K-1}\rightarrow\overline{\phi}} \leq 0
    \label{AchRegion_FME_0_positiveSplit_8}
\end{align}

The problem is to project the above nine-dimensional polytope in original- and split-rates space (described by the inequalities \eqref{AchRegion_FME_0_j_1}-\eqref{AchRegion_FME_0_K_2} and \eqref{AchRegion_FME_0_positiveSplit_1}-\eqref{AchRegion_FME_0_positiveSplit_8}) onto the space of the four original rates using the FME procedure. The size of the problem is large in that there are a large number of (groups of) inequalities and as many as five split rates to be eliminated. Performing the FME procedure correctly for a problem of this size in an error-free manner can be an arduous task. In particular, the order in which the split rates are eliminated appears to have a strong bearing on whether the task can be completed at all, with reasonable effort. We propose an order that we were not only able to use to perform the projection to its completion but that also provides opportunities for identifying errors in the process early enough and hence correcting them before proceeding further. In particular, we propose to do the projection in five steps in the following order:
\begin{enumerate}
    \item $R_{ 
\overline{K-1} \rightarrow \overline{\phi}  }$
    \item $R_{ 
\overline{K} \rightarrow \overline{\phi}  }$
    \item $R_{ 
\overline{K-1.K} \rightarrow \overline{K}}  $
    \item $R_{ 
\overline{K-1.K} \rightarrow \overline{K-1}}  $
    \item $R_{ 
\overline{K-1.K} \rightarrow \overline{\phi}  }$
\end{enumerate}
Note that this order maintains a certain symmetric structure in the problem, in that after projecting away the first two split rates in the above list we expect a symmetry in the description of the resulting seven-dimensional polytope around $K$ and $K-1$ (i.e., exchanging $K$ and $K-1$ should not change it). Similarly, such a symmetry should result again after projecting away the first four split rates in the above list, 
and eventually after projecting away the fifth split rate that yields the four-dimensional polytopes specified in Theorem  \ref{Th_AchRegion_FourMsgs}, which is indeed symmetric in the sense described above.

Note also that the number of inequalities in \eqref{AchRegion_FME_0_j_1}-\eqref{AchRegion_FME_0_j_5} is indeterminate since it is function of $K$. Hence, we deal with these groups of inequalities at once, 
instead of exhaustively listing them. 
Moreover, at the end of each projection to eliminate a split rate, it is important to remove all redundant inequalities so that the FME procedure remains tractable in future steps. 
We do not address the issue of why inequalities that are not retained at the end of each step are redundant, leaving this task at each step (beyond step 2, when they arise) to the interested reader. Judicious applications of the data processing inequality and the Markov chain relationships between the involved random variables are all that are needed. The 
task of identifying which inequalities are redundant is, however, solved for the reader at each step in that these are precisely the inequalities that we omit from including in the projected reduced-dimensional region given at the end of each step.


In the first two steps, we project away the split rates $R_{\overline{K-1}\rightarrow \overline{\phi}}$ and $R_{\overline{K}\rightarrow\overline{\phi}}$
and obtain, after rearranging the inequalities, the following inequalities with $j\in\{1,2,\cdots, K-2\}$
\begin{align}
    &R_{\overline{\phi}}+R_{\overline{K-1}}+R_{\overline{K}}+R_{\overline{K-1.K}} \nonumber \\
    &\quad \leq I(X;Y_j|U_{\overline{\phi}},U_{\overline{K}})+I(U_{\overline{\phi}},U_{\overline{K}};Y_{K-1})
    \label{AchRegion_FME_2_1}\\
    &R_{\overline{\phi}}+R_{\overline{K-1}}+R_{\overline{K}}+R_{\overline{K-1.K}} \nonumber \\
    &\quad \leq I(X;Y_j|U_{\overline{\phi}},U_{\overline{K-1}})+I(U_{\overline{\phi}},U_{\overline{K-1}};Y_{K})
    \label{AchRegion_FME_2_2}\\
    &R_{\overline{\phi}}+R_{\overline{K-1}}+R_{\overline{K}}+R_{\overline{K-1.K}} \leq 
    I(X;Y_j) 
    \label{AchRegion_FME_2_3}\\
    & R_{\overline{K-1.K}}-R_{\overline{K-1.K}\rightarrow \overline{\phi}} \leq 
    I(X;Y_j|U_{\overline{\phi}})
    \label{AchRegion_FME_2_4}\\
    &R_{\overline{K-1.K}}-R_{\overline{K-1.K}\rightarrow \overline{\phi}}-R_{\overline{K-1.K}\rightarrow \overline{K-1}}\nonumber \\
    &\quad \leq 
     I(X;Y_j|U_{\overline{\phi}},U_{\overline{K-1}})
    \label{AchRegion_FME_2_5}\\
    &R_{\overline{K-1.K}}-R_{\overline{K-1.K}\rightarrow \overline{\phi}}-R_{\overline{K-1.K}\rightarrow \overline{K}} \leq 
     I(X;Y_j|U_{\overline{\phi}},U_{\overline{K}})
    \label{AchRegion_FME_2_6}\\
    \nonumber 
    \\
    & R_{\overline{K-1.K}\rightarrow \overline{K-1}} \leq  I(U_{\overline{K-1}};Y_{K}|U_{\overline{\phi}})
    \label{AchRegion_FME_2_7}\\
    & R_{\overline{K-1.K}\rightarrow \overline{K}} \leq  I(U_{\overline{K}};Y_{K-1}|U_{\overline{\phi}})
    \label{AchRegion_FME_2_8}\\
    \nonumber 
    \\
    &R_{\overline{\phi}}+R_{\overline{K-1}} +R_{\overline{K-1.K} \rightarrow \overline{\phi}}+R_{\overline{K-1.K}\rightarrow \overline{K-1}}\nonumber \\
    &\quad \leq 
     I(U_{\overline{\phi}},U_{\overline{K-1}};Y_{K})
    \label{AchRegion_FME_2_9}\\
    &R_{\overline{\phi}}+R_{\overline{K}} +R_{\overline{K-1.K} \rightarrow \overline{\phi}}+R_{\overline{K-1.K}\rightarrow \overline{K}}\nonumber \\
    &\quad \leq 
     I(U_{\overline{\phi}},U_{\overline{K}};Y_{K-1})
    \label{AchRegion_FME_2_10}\\
    \nonumber 
    \\
    &R_{\overline{\phi}}+R_{\overline{K-1}} +R_{\overline{K}}+ R_{\overline{K-1.K}} +R_{\overline{K-1.K}\rightarrow \overline{K-1}}\nonumber \\
    &\quad \leq 
     I(X;Y_j|U_{\overline{\phi}})+ I(U_{\overline{\phi}},U_{\overline{K-1}};Y_{K})
    \label{AchRegion_FME_2_11}\\
    &R_{\overline{\phi}}+R_{\overline{K-1}} +R_{\overline{K}}+ R_{\overline{K-1.K}} +R_{\overline{K-1.K}\rightarrow \overline{K}}\nonumber \\
    &\quad \leq 
     I(X;Y_j|U_{\overline{\phi}})+ I(U_{\overline{\phi}},U_{\overline{K}};Y_{K-1})
    \label{AchRegion_FME_2_12}\\
    \nonumber 
    \\
     &R_{\overline{\phi}}+R_{\overline{K-1}} +R_{\overline{K}}+ 
     \nonumber \\ 
     & \quad  R_{\overline{K-1.K}\rightarrow \overline{\phi}}+ R_{\overline{K-1.K}\rightarrow \overline{K-1}} +R_{\overline{K-1.K}\rightarrow \overline{K}}\nonumber \\
    &\quad \leq 
      I(U_{\overline{K}};Y_{K-1}|U_{\overline{\phi}})+ I(U_{\overline{\phi}},U_{\overline{K-1}};Y_{K})
    \label{AchRegion_FME_2_13}\\
    &R_{\overline{\phi}}+R_{\overline{K-1}} +R_{\overline{K}}+ 
     \nonumber \\ 
     & \quad  R_{\overline{K-1.K}\rightarrow \overline{\phi}}+ R_{\overline{K-1.K}\rightarrow \overline{K-1}} +R_{\overline{K-1.K}\rightarrow \overline{K}}\nonumber \\
    &\quad \leq 
      I(U_{\overline{K-1}};Y_{K}|U_{\overline{\phi}})+ I(U_{\overline{\phi}},U_{\overline{K}};Y_{K-1})
    \label{AchRegion_FME_2_14}\\
    \nonumber 
    \\
    &R_{\overline{K-1.K}}-R_{\overline{K-1.K}\rightarrow \overline{\phi}}- R_{\overline{K-1.K}\rightarrow \overline{K-1}} -R_{\overline{K-1.K}\rightarrow \overline{K}}\nonumber \\
    &\quad \leq 
    I(X;Y_j|U_{\overline{\phi}},U_{\overline{K-1}},U_{\overline{K}})
    \nonumber 
    \\
    \label{AchRegion_FME_2_15}\\
    &2R_{\overline{\phi}}+2R_{\overline{K-1}} +2R_{\overline{K}}+ \nonumber \\ & \quad R_{\overline{K-1.K}}+R_{\overline{K-1.K}\rightarrow \overline{\phi}}+ R_{\overline{K-1.K}\rightarrow \overline{K-1}} +R_{\overline{K-1.K}\rightarrow \overline{K}}\nonumber \\
    &\quad \leq I(X;Y_j|U_{\overline{\phi}}) +I(U_{\overline{\phi}},U_{\overline{K}};Y_{K-1})+I(U_{\overline{\phi}},U_{\overline{K-1}};Y_{K})
    \label{AchRegion_FME_2_16}\\
     \nonumber 
     \\
        & -R_{\overline{K-1.K}}+ \nonumber \\
    & \qquad R_{\overline{K-1.K}\rightarrow\overline{K-1}}+R_{\overline{K-1.K}\rightarrow\overline{K}}+R_{\overline{K-1.K}\rightarrow\overline{\phi}}\leq 0
    \label{AchRegion_FME_2_17}\\
    &-R_{\overline{K-1.K}\rightarrow\overline{K-1}} \leq 0
    \label{AchRegion_FME_2_18}\\
    &-R_{\overline{K-1.K}\rightarrow\overline{K}} \leq 0
    \label{AchRegion_FME_2_19}\\
    &-R_{\overline{K-1.K}\rightarrow\overline{\phi}} \leq 0
    \label{AchRegion_FME_2_20}
\end{align}

Observe that the above seven-dimensional polytope remains unchanged if $K-1$ and $K$ are exchanged, as one might expect.
Indeed, observe this symmetry in every group of inequalities separated by a new line.  There were no redundant inequalities in the first two steps.

In the third step, we project away the split rate $R_{\overline{K-1.K}\rightarrow\overline{K}}$.
Notice now that we have a group of inequalities in \eqref{AchRegion_FME_2_15} which give lower bounds for $R_{\overline{K-1.K}\rightarrow\overline{K}}$ and a group of inequalities in \eqref{AchRegion_FME_2_12} and \eqref{AchRegion_FME_2_16} that give upper bounds for $R_{\overline{K-1.K}\rightarrow\overline{K}}$. Hence, to project the split rate $R_{\overline{K-1.K}\rightarrow\overline{K}}$ from these, we have to introduce new variables $j_1,j_2\in\{1,2,\cdots,K-2\}$ since we get a multiplicative increase in the number of inequalities.
By projecting away the split rate $R_{\overline{K-1.K}\rightarrow\overline{K}}$, we obtain, by rearranging the inequalities and getting rid of the redundant 
ones
\footnote{Note that we get redundant inequalities by eliminating the split rate $R_{\overline{K-1.K}\rightarrow\overline{K}}$ from \eqref{AchRegion_FME_2_6} and \eqref{AchRegion_FME_2_10}, \eqref{AchRegion_FME_2_12}, \eqref{AchRegion_FME_2_14}, and \eqref{AchRegion_FME_2_16}, and also from \eqref{AchRegion_FME_2_19} and \eqref{AchRegion_FME_2_12}.} 
, the following for $j,j_1,j_2\in\{1,2,\cdots,K-2\}$
\begin{align}
    &R_{\overline{\phi}}+R_{\overline{K-1}}+R_{\overline{K}}+R_{\overline{K-1.K}} \nonumber \\
    &\quad \leq I(X;Y_j|U_{\overline{\phi}},U_{\overline{K}})+I(U_{\overline{\phi}},U_{\overline{K}};Y_{K-1})
    \label{AchRegion_FME_3_1}\\
    &R_{\overline{\phi}}+R_{\overline{K-1}}+R_{\overline{K}}+R_{\overline{K-1.K}} \nonumber \\
    &\quad \leq I(X;Y_j|U_{\overline{\phi}},U_{\overline{K-1}})+I(U_{\overline{\phi}},U_{\overline{K-1}};Y_{K})
    \label{AchRegion_FME_3_2}\\
    &R_{\overline{\phi}}+R_{\overline{K-1}}+R_{\overline{K}}+R_{\overline{K-1.K}}  \leq 
    I(X;Y_j) 
    \label{AchRegion_FME_3_3}\\
    &R_{\overline{\phi}}+ R_{\overline{K-1}}+ R_{\overline{K}} +R_{\overline{K-1.K}} 
    \nonumber \\ 
    & \quad \leq 
    I(X;Y_j|U_{\overline{\phi}},U_{\overline{K-1}},U_{\overline{K}}) \nonumber \\ 
    & \quad +I(U_{\overline{K}};Y_{K-1}|U_{\overline{\phi}}) +I(U_{\overline{\phi}},U_{\overline{K-1}};Y_{K})
    \label{AchRegion_FME_3_4}\\
    &R_{\overline{\phi}}+ R_{\overline{K-1}}+ R_{\overline{K}} +R_{\overline{K-1.K}} \nonumber \\
    & \quad \leq  
    I(X;Y_j|U_{\overline{\phi}},U_{\overline{K-1}},U_{\overline{K}}) \nonumber \\ & 
    \quad +I(U_{\overline{K-1}};Y_{K}|U_{\overline{\phi}}) +I(U_{\overline{\phi}},U_{\overline{K}};Y_{K-1})
    \label{AchRegion_FME_3_5}\\
    &2R_{\overline{\phi}}+ 2R_{\overline{K-1}}+ 2R_{\overline{K}} +2R_{\overline{K-1.K}} \nonumber \\
    & \quad \leq  
    I(X;Y_{j_1}|U_{\overline{\phi}}) 
    +   I(X;Y_{j_2}|U_{\overline{\phi}},U_{\overline{K-1}},U_{\overline{K}})
    \nonumber \\ & \quad
    +I(U_{\overline{\phi}},U_{\overline{K}};Y_{K-1})+I(U_{\overline{\phi}},U_{\overline{K-1}};Y_{K})
    \label{AchRegion_FME_3_6}\\
    &R_{\overline{K-1.K}}-R_{\overline{K-1.K} \rightarrow \overline{\phi}} \leq I(X;Y_j|U_{\overline{\phi}})
    \label{AchRegion_FME_3_7}\\
    &R_{\overline{K-1.K}}-R_{\overline{K-1.K} \rightarrow \overline{\phi}}\nonumber \\
    &\quad \leq I(X;Y_j|U_{\overline{\phi}},U_{\overline{K}}) + I(U_{\overline{K}};Y_{K-1}|U_{\overline{\phi}})
    \label{AchRegion_FME_3_8}\\
    &R_{\overline{K-1.K}}-R_{\overline{K-1.K} \rightarrow \overline{\phi}}-R_{\overline{K-1.K} \rightarrow \overline{K-1}} \nonumber \\
    &\quad \leq I(X;Y_j|U_{\overline{\phi}},U_{\overline{K-1}})
    \label{AchRegion_FME_3_9}\\
    & R_{\overline{K-1.K}}-R_{\overline{K-1.K} \rightarrow \overline{\phi}}-R_{\overline{K-1.K} \rightarrow \overline{K-1}} \nonumber \\
    &\quad \leq 
    I(X;Y_j|U_{\overline{\phi}},U_{\overline{K-1}},U_{\overline{K}})+ I(U_{\overline{K}};Y_{K-1}|U_{\overline{\phi}})
     \label{AchRegion_FME_3_10}\\
    & R_{\overline{\phi}} +R_{\overline{K}}+R_{\overline{K-1.K}}-R_{\overline{K-1.K} \rightarrow \overline{K-1}} \nonumber \\
    & \quad \leq 
     I(X;Y_j|U_{\overline{\phi}},U_{\overline{K-1}},U_{\overline{K}})+ I(U_{\overline{\phi}},U_{\overline{K}};Y_{K-1})
    \label{AchRegion_FME_3_11}\\
    & R_{\overline{\phi}} +R_{\overline{K-1}}+R_{\overline{K}}+  2R_{\overline{K-1.K}}-R_{\overline{K-1.K} \rightarrow \overline{\phi}}- \nonumber \\   & \quad R_{\overline{K-1.K} \rightarrow \overline{K-1}}  \leq I(X;Y_{j_1}|U_{\overline{\phi}})+ \nonumber \\ 
    & \quad  I(X;Y_{j_2}|U_{\overline{\phi}},U_{\overline{K-1}},U_{\overline{K}})
    +I(U_{\overline{\phi}},U_{\overline{K}};Y_{K-1})
    \label{AchRegion_FME_3_12}\\
    & R_{\overline{K-1.K} \rightarrow \overline{K-1}} \leq  I(X;Y_j|U_{\overline{\phi}},U_{\overline{K-1}})
    \label{AchRegion_FME_3_13}\\
    & R_{\overline{K-1.K} \rightarrow \overline{K-1}} \leq   I(U_{\overline{K-1}};Y_{K}|U_{\overline{\phi}})
    \label{AchRegion_FME_3_14}\\
    & R_{\overline{\phi}} +R_{\overline{K}}+ R_{\overline{K-1.K}\rightarrow \overline{\phi}} \leq 
     I(U_{\overline{\phi}},U_{\overline{K}};Y_{K-1})
    \label{AchRegion_FME_3_15}\\
    &R_{\overline{\phi}} +R_{\overline{K-1}}+R_{\overline{K}}+R_{\overline{K-1.K}} +R_{\overline{K-1.K}\rightarrow \overline{K-1}}\nonumber \\
    &\quad \leq 
    I(X;Y_j|U_{\overline{\phi}})+ I(U_{\overline{\phi}},U_{\overline{K-1}};Y_{K})
    \label{AchRegion_FME_3_16}\\
    &R_{\overline{\phi}} +R_{\overline{K-1}}+R_{\overline{K}}+R_{\overline{K-1.K}} +R_{\overline{K-1.K}\rightarrow \overline{K-1}}\nonumber \\
    &\quad \leq I(X;Y_j|U_{\overline{\phi}},U_{\overline{K}})\nonumber \\
    & \quad + I(U_{\overline{K}};Y_{K-1}|U_{\overline{\phi}})+
    I(U_{\overline{\phi}},U_{\overline{K-1}};Y_{K})
    \label{AchRegion_FME_3_17}\\
    &R_{\overline{\phi}} +R_{\overline{K-1}}+R_{\overline{K-1.K}\rightarrow \overline{\phi}} +R_{\overline{K-1.K}\rightarrow \overline{K-1}}\nonumber \\
    &\quad \leq I(U_{\overline{\phi}},U_{\overline{K-1}};Y_{K})
    \label{AchRegion_FME_3_18}\\
    &R_{\overline{\phi}} +R_{\overline{K-1}}+R_{\overline{K}} +R_{\overline{K-1.K}\rightarrow \overline{\phi}} +R_{\overline{K-1.K}\rightarrow \overline{K-1}}\nonumber \\
    &\quad \leq I(U_{\overline{K}};Y_{K-1}|U_{\overline{\phi}})+ I(U_{\overline{\phi}},U_{\overline{K-1}};Y_{K})
    \label{AchRegion_FME_3_19}\\
    &R_{\overline{\phi}} +R_{\overline{K-1}}+R_{\overline{K}} +R_{\overline{K-1.K}\rightarrow \overline{\phi}} +R_{\overline{K-1.K}\rightarrow \overline{K-1}}\nonumber \\
    &\quad \leq I(U_{\overline{K-1}};Y_{K}|U_{\overline{\phi}})+ I(U_{\overline{\phi}},U_{\overline{K}};Y_{K-1})
    \label{AchRegion_FME_3_20}\\
    &2R_{\overline{\phi}} +2R_{\overline{K-1}}+2R_{\overline{K}} +R_{\overline{K-1.K}} +R_{\overline{K-1.K}\rightarrow \overline{\phi}} \nonumber \\
    &\quad +R_{\overline{K-1.K}\rightarrow \overline{K-1}} \leq 
    I(X;Y_j|U_{\overline{\phi}}) \nonumber \\ 
    & \quad +I(U_{\overline{\phi}},U_{\overline{K}};Y_{K-1})+ I(U_{\overline{\phi}},U_{\overline{K-1}};Y_{K})
    \label{AchRegion_FME_3_21}\\
    & -R_{\overline{K-1.K}}+  R_{\overline{K-1.K}\rightarrow\overline{K-1}}+R_{\overline{K-1.K}\rightarrow\overline{\phi}}\leq 0
    \label{AchRegion_FME_3_22}\\
    &-R_{\overline{K-1.K}\rightarrow\overline{K-1}} \leq 0
    \label{AchRegion_FME_3_23}\\
    &-R_{\overline{K-1.K}\rightarrow\overline{\phi}} \leq 0
    \label{AchRegion_FME_3_24}
\end{align}

In the fourth step, we project away the split rate $R_{\overline{K-1.K}\rightarrow\overline{K-1}}$ from the polytope described by the inequalities \eqref{AchRegion_FME_3_1}-\eqref{AchRegion_FME_3_24}. Note that the inequalities \eqref{AchRegion_FME_3_9}-\eqref{AchRegion_FME_3_12} and \eqref{AchRegion_FME_3_23} give lower bounds and the inequalities  \eqref{AchRegion_FME_3_13}-\eqref{AchRegion_FME_3_14}, and \eqref{AchRegion_FME_3_16}-\eqref{AchRegion_FME_3_22} give upper bounds on $R_{\overline{K-1.K}\rightarrow\overline{K-1}}$. Hence, we have $K^2{-}K{-}1$ lower bounds and $4K{-}3$ upper bounds with the rest of the inequalities remaining unchanged.

Projecting away the split rate $R_{\overline{K-1.K}\rightarrow\overline{K-1}}$ from the aforementioned inequalities, in which there are five lower bounds and nine upper bounds, counting a group of inequalities as one inequality, we obtain 45 groups of inequalities (including groups of size one). However, and we leave it to the reader to show this, the nine inequalities, besides the unchanged inequalities, shown next, are the only ones that are not redundant. Hence, we get the five-dimensional polytope described by the following inequalities for all $j,j_1,j_2\in\{1,2,\cdots,K-2\}$
\begin{align}
    &R_{\overline{\phi}}+R_{\overline{K-1}}+R_{\overline{K}}+R_{\overline{K-1.K}} \nonumber \\
    &\quad \leq I(X;Y_j|U_{\overline{\phi}},U_{\overline{K}})+I(U_{\overline{\phi}},U_{\overline{K}};Y_{K-1})
    \label{AchRegion_FME_4_1}\\
    &R_{\overline{\phi}}+R_{\overline{K-1}}+R_{\overline{K}}+R_{\overline{K-1.K}} \nonumber \\
    &\quad \leq I(X;Y_j|U_{\overline{\phi}},U_{\overline{K-1}})+I(U_{\overline{\phi}},U_{\overline{K-1}};Y_{K})
    \label{AchRegion_FME_4_2}\\
    \nonumber 
    \\
    &R_{\overline{\phi}}+R_{\overline{K-1}}+R_{\overline{K}}+R_{\overline{K-1.K}} 
    \leq 
    I(X;Y_j) 
    \label{AchRegion_FME_4_3}\\
    &R_{\overline{\phi}}+ R_{\overline{K-1}}+ R_{\overline{K}} +R_{\overline{K-1.K}} 
    \nonumber \\ 
    & \quad \leq 
    I(X;Y_j|U_{\overline{\phi}},U_{\overline{K-1}},U_{\overline{K}}) \nonumber \\ 
    & \quad +I(U_{\overline{K}};Y_{K-1}|U_{\overline{\phi}}) +I(U_{\overline{\phi}},U_{\overline{K-1}};Y_{K})
    \label{AchRegion_FME_4_4}\\
    \nonumber 
    \\
    &R_{\overline{\phi}}+ R_{\overline{K-1}}+ R_{\overline{K}} +R_{\overline{K-1.K}} \nonumber \\
    & \quad \leq  
    I(X;Y_j|U_{\overline{\phi}},U_{\overline{K-1}},U_{\overline{K}}) \nonumber \\ & 
    \quad +I(U_{\overline{K-1}};Y_{K}|U_{\overline{\phi}}) +I(U_{\overline{\phi}},U_{\overline{K}};Y_{K-1})
    \label{AchRegion_FME_4_5}\\
    &2R_{\overline{\phi}}+ 2R_{\overline{K-1}}+ 2R_{\overline{K}} +2R_{\overline{K-1.K}} \nonumber \\
    & \quad \leq  
    I(X;Y_{j_1}|U_{\overline{\phi}}) 
    +   I(X;Y_{j_2}|U_{\overline{\phi}},U_{\overline{K-1}},U_{\overline{K}})
    \nonumber \\ & \quad
    +I(U_{\overline{\phi}},U_{\overline{K}};Y_{K-1})+I(U_{\overline{\phi}},U_{\overline{K-1}};Y_{K})
    \label{AchRegion_FME_4_6}\\
    \nonumber 
    \\
    &R_{\overline{\phi}}+R_{\overline{K-1}}+R_{\overline{K-1.K}\rightarrow \overline{\phi}}
    \leq I(U_{\overline{\phi}},U_{\overline{K-1}};Y_{K})
    \label{AchRegion_FME_4_7}\\
    &R_{\overline{\phi}}+R_{\overline{K-1}}+R_{\overline{K-1.K}\rightarrow \overline{\phi}}
    \leq I(U_{\overline{\phi}},U_{\overline{K}};Y_{K-1})
    \label{AchRegion_FME_4_8}\\
    \nonumber 
    \\
    &R_{\overline{\phi}}+R_{\overline{K-1}}+R_{\overline{K}}+R_{\overline{K-1.K}\rightarrow \overline{\phi}}\nonumber \\
    &\quad \leq I(U_{\overline{K}};Y_{K-1}|U_{\overline{\phi}})+ I(U_{\overline{\phi}},U_{\overline{K-1}};Y_{K})
    \label{AchRegion_FME_4_9}\\
    &R_{\overline{\phi}}+R_{\overline{K-1}}+R_{\overline{K}}+R_{\overline{K-1.K}\rightarrow \overline{\phi}}\nonumber \\
    &\quad \leq I(U_{\overline{K-1}};Y_{K}|U_{\overline{\phi}})+
     I(U_{\overline{\phi}},U_{\overline{K}};Y_{K-1})
    \label{AchRegion_FME_4_10}\\
    \nonumber 
    \\
    &2R_{\overline{\phi}}+R_{\overline{K-1}}+R_{\overline{K}}
    +R_{\overline{K-1.K}}+R_{\overline{K-1.K}\rightarrow \overline{\phi}}\nonumber \\
    &\quad \leq I(X;Y_j|U_{\overline{\phi}},U_{\overline{K-1}},U_{\overline{K}})\nonumber \\
    &\quad +I(U_{\overline{\phi}},U_{\overline{K}};Y_{K-1})+I(U_{\overline{\phi}},U_{\overline{K-1}};Y_{K})
    \label{AchRegion_FME_4_11}\\
    &2R_{\overline{\phi}}+2R_{\overline{K-1}}+2R_{\overline{K}}
    +R_{\overline{K-1.K}}+R_{\overline{K-1.K}\rightarrow \overline{\phi}}
    \nonumber \\
    &\quad 
    \leq I(X;Y_j|U_{\overline{\phi}})
    +I(U_{\overline{\phi}},U_{\overline{K}};Y_{K-1})+I(U_{\overline{\phi}},U_{\overline{K-1}};Y_{K})
    \label{AchRegion_FME_4_12}\\
    \nonumber 
    \\
    &R_{\overline{K-1.K}}-R_{\overline{K-1.K}\rightarrow \overline{\phi}} \leq
    I(X;Y_j|U_{\overline{\phi}})
     \label{AchRegion_FME_4_13}\\
    &R_{\overline{K-1.K}}-R_{\overline{K-1.K}\rightarrow \overline{\phi}}\nonumber \\
    &\quad \leq
     I(X;Y_{j_1}|U_{\overline{\phi}},U_{\overline{K-1}})
    +I(X;Y_{j_2}|U_{\overline{\phi}},U_{\overline{K}})
     \label{AchRegion_FME_4_14}\\
    &R_{\overline{K-1.K}}-R_{\overline{K-1.K}\rightarrow \overline{\phi}}\nonumber \\
    &\quad \leq
    I(X;Y_j|U_{\overline{\phi}},U_{\overline{K-1}})+I(U_{\overline{K-1}};Y_{K}|U_{\overline{\phi}})
     \label{AchRegion_FME_4_15}\\
    &R_{\overline{K-1.K}}-R_{\overline{K-1.K}\rightarrow \overline{\phi}}\nonumber \\
    &\quad \leq
    I(X;Y_j|U_{\overline{\phi}},U_{\overline{K}})+I(U_{\overline{K}};Y_{K-1}|U_{\overline{\phi}})
     \label{AchRegion_FME_4_16}\\
    &R_{\overline{K-1.K}}-R_{\overline{K-1.K}\rightarrow \overline{\phi}}\nonumber \\
    &\quad \leq
    I(X;Y_j|U_{\overline{\phi}},U_{\overline{K-1}},U_{\overline{K}})\nonumber \\ &\quad +I(U_{\overline{K}};Y_{K-1}|U_{\overline{\phi}})+I(U_{\overline{K-1}};Y_{K}|U_{\overline{\phi}})
     \label{AchRegion_FME_4_17}\\
     \nonumber 
     \\
     & -R_{\overline{K-1.K}}+R_{\overline{K-1.K}\rightarrow\overline{\phi}}\leq 0
    \label{AchRegion_FME_4_18}\\
    &-R_{\overline{K-1.K}\rightarrow\overline{\phi}} \leq 0
    \label{AchRegion_FME_4_19}
\end{align}

Observe again the symmetry in the above five-dimensional polytope, which is that it remains unchanged when $K-1$ and $K$ are exchanged
(even within every group of inequalities separated by new line), as expected.  

Finally, in the fifth step, we project away the split rate $R_{\overline{K-1.K}\rightarrow\overline{\phi}}$ from the five-dimensional polytope described by the inequalities \eqref{AchRegion_FME_4_1}-\eqref{AchRegion_FME_4_19}. 
After doing this projection, and eliminating redundant inequalities, it is left to the reader to verify that we obtain the four-dimensional polytope in the original rates given in Theorem \ref{Th_AchRegion_FourMsgs}. Note again the presence of symmetry of that polytope when $K-1$ and $K$ are exchanged in that four-dimensional polytope. This concludes the outline of the FME procedure.

\section{Proof of Theorem \ref{thm:wbinning}}
\label{app:thmbinning}
The achievable scheme of Theorem \ref{Th_AchRegion_FourMsgs} that is based on rate-splitting and superposition coding is enhanced by adding to it the technique of binning
 that involves generating the $U_{\overline{K}}$ and $U_{\overline{K-1}}$ codebooks at excess rate and creating bins of codewords in place of codewords for the associated messages and selecting a pair of codewords from the product bins selected by the messages that are jointly typical per the allowed joint distribution of $U_{\overline{K}}$ and $U_{\overline{K-1}}$ (conditioned on $U_{\overline{\phi}}$). In particular, the messages $M_{\overline{K}}$, $M_{\overline{K-1}}$, and $M_{\overline{K-1.K}}$ are split as in the proof of Theorem \ref{Th_AchRegion_FourMsgs}. 
The reconstructed message $(M_{\overline{\phi}},M_{\overline{K-1} \rightarrow \overline{\phi}}, M_{\overline{K} \rightarrow \overline{\phi}},M_{\overline{K-1.K} \rightarrow \overline{\phi}} )$ is represented by the cloud center $U_{\overline{\phi}}$ again as in Theorem \ref{Th_AchRegion_FourMsgs}. 
For each $u_{\overline{\phi}}^n(m_{\overline{\phi}},m_{\overline{K-1} \rightarrow \overline{\phi}}, m_{\overline{K} \rightarrow \overline{\phi}},m_{\overline{K-1.K} \rightarrow \overline{\phi}} )$ randomly and independently generate (a) $2^{n \tilde{R}_{\overline{K}}}$ independent sequences $u_{\overline{K}}^n(m_{\overline{\phi}},m_{\overline{K-1} \rightarrow \overline{\phi}}, m_{\overline{K} \rightarrow \overline{\phi}},m_{\overline{K-1.K} \rightarrow \overline{\phi}}, \tilde{m}_{\overline{K}})$, for each $\tilde{m}_{\overline{K}} \in[1:2^{n\tilde{R}_{\overline{K}}}]$ with 
\begin{eqnarray}
\label{eq:excess1}
\tilde{R}_{\overline{K}}\geq R_{\overline{K-1.K} \rightarrow \overline{K}}+R_{\overline{K}\rightarrow \overline{K}}
\end{eqnarray}
and (b) $2^{n \tilde{R}_{\overline{K-1}}}$ independent sequences $u_{\overline{K-1}}^n(m_{\overline{\phi}},m_{\overline{K-1} \rightarrow \overline{\phi}}, m_{\overline{K} \rightarrow \overline{\phi}},m_{\overline{K-1.K} \rightarrow \overline{\phi}}, \tilde{m}_{\overline{K-1}})$ for each $\tilde{m}_{\overline{K-1}} \in[1:2^{n\tilde{R}_{\overline{K-1}}}]$ and
\begin{eqnarray} 
\label{eq:excess2}
\tilde{R}_{\overline{K-1}}\geq R_{\overline{K-1.K} \rightarrow \overline{K-1}}+R_{\overline{K-1}\rightarrow \overline{K-1}}
\end{eqnarray} 
Partition the $2^{n \tilde{R}_{\overline{K}}}$ sequences $u_{\overline{K}}^n$ into $2^{n(R_{\overline{K-1.K} \rightarrow \overline{K}}+R_{\overline{K}\rightarrow \overline{K}})}$ equi-sized subcodebooks (or bins) and the $2^{n \tilde{R}_{\overline{K-1}}}$ sequences $u_{\overline{K-1}}^n$ into $2^{n(R_{\overline{K-1.K} \rightarrow \overline{K-1}}+R_{\overline{K-1}\rightarrow \overline{K-1}})}$ equi-sized subcodebooks,  with the subcodebooks indexed by the sub-message pairs $(m_{\overline{K-1.K} \rightarrow \overline{K}}, m_{\overline{K}\rightarrow \overline{K}})$ and $(m_{\overline{K-1.K} \rightarrow \overline{K-1}}, m_{\overline{K-1}\rightarrow \overline{K-1}})$, respectively. 
The tuple $(m_{\overline{\phi}}$ $,m_{\overline{K-1} \rightarrow \overline{\phi}},$ $ m_{\overline{K} \rightarrow \overline{\phi}},$ $m_{\overline{K-1.K} \rightarrow \overline{\phi}},$ $m_{\overline{K-1.K} \rightarrow \overline{K}},$ $m_{\overline{K}\rightarrow \overline{K}},$ $m_{\overline{K-1.K} \rightarrow \overline{K-1}},$ $ m_{\overline{K-1}\rightarrow \overline{K-1}})$ identifies two subcodebooks of $u_{\overline{K}}^n$ and $u_{\overline{K-1}}^n$ sequences respectively from which a jointly typical ($u_{\overline{K}}^n$, $u_{\overline{K-1}}^n$ ) pair is selected. 
For each such pair, generate $2^{nR_{\overline{K-1.K} \rightarrow\overline{K-1.K} }}$ independent $x^n$ sequences indexed by the sub-message $m_{\overline{K-1.K} \rightarrow\overline{K-1.K}}$. To ensure that each product bin contains a jointly typical pair with arbitrarily high probability, we require by the mutual covering lemma \cite{el2011network} that
\begin{eqnarray}
    R_{\overline{K-1.K} \rightarrow \overline{K}} &+R_{\overline{K}\rightarrow \overline{K}}+R_{\overline{K-1.K} \rightarrow \overline{K-1}}+R_{\overline{K-1}\rightarrow \overline{K-1}} \nonumber \\ & \leq \tilde{R}_{\overline{K}}+\tilde{R}_{\overline{K-1}}-I(U_{\overline{K}} ;U_{\overline{K-1}}|U_{\overline{\phi}} )
    \label{Eq_Binning_cond}
\end{eqnarray}

Receivers $\{Y_j\}_{j=1}^{K-2}$ decode all four intended messages by the joint unique decoding of $X$ (and hence of $U_{\overline{\phi}},U_{\overline{K}},U_{\overline{K-1}},U_{\overline{K-1.K}}$), and this happens successfully as long as
\eqref{AchRegion_FME_0_j_1}-\eqref{AchRegion_FME_0_j_5} hold.
On the other hand, receiver $Y_{K-1}$ finds its two intended messages $M_{\overline{\phi}}$ and $ M_{\overline{K}}$ by jointly uniquely decoding $(U_{\overline{\phi}},U_{\overline{K}})$, and similarly, receiver $Y_{K}$ find its two intended messages $M_{\overline{\phi}}$ and $M_{\overline{K-1}}$ by jointly uniquely decoding  $(U_{\overline{\phi}},U_{\overline{K-1}})$. Receivers $Y_{K-1}$ and $Y_{K}$ decode their intended pair of messages successfully provided the following inequalities hold:
\begin{align}
    &R_{\overline{\phi}}+R_{\overline{K} \rightarrow \overline{\phi}}
     +R_{\overline{K-1} \rightarrow \overline{\phi} } 
    +R_{\overline{K-1.K} \rightarrow \overline{\phi} }+\tilde{R}_{\overline{K}}
    \nonumber \\ & \hspace{1cm}
    \leq I(U_{\overline{\phi}},U_{\overline{K}};Y_{K-1})
    \label{AchRegion_FME_0_Binning_K-1_1}\\
    &\tilde{R}_{\overline{K}}\leq I(U_{\overline{K}};Y_{K-1}|U_{\overline{\phi}})
    \label{AchRegion_FME_0_Binning_K-1_2}\\
    &R_{\overline{\phi}}+R_{\overline{K}\rightarrow \overline{\phi}}
     +R_{\overline{K-1} \rightarrow \overline{\phi} } +R_{\overline{K-1.K} \rightarrow \overline{\phi} }+\tilde{R}_{\overline{K-1}}
     \nonumber \\  \hspace{1cm}
     &\leq I(U_{\overline{\phi}},U_{\overline{K-1}};Y_{K})
    \label{AchRegion_FME_0_Binning_K_1}\\
    &\tilde{R}_{\overline{K-1}}\leq I(U_{\overline{K-1}};Y_{K}|U_{\overline{\phi}})
    \label{AchRegion_FME_0_Binning_K_2}
\end{align}
At this point, inequalities \eqref{AchRegion_FME_0_j_1}-\eqref{AchRegion_FME_0_j_5}, \eqref{eq:excess1}- 
\eqref{AchRegion_FME_0_Binning_K_2} and the non-negativity of message rate constraints define a 11-dimensional polytope. Extending the FME method as outlined in Appendix \ref{Appendix_FME_TH1} for the projection required in Theorem \ref{Th_AchRegion_FourMsgs} \textemdash the details of which are left to the reader \textemdash we obtain the rate region given in the statement of Theorem \ref{thm:wbinning}. 

\bibliographystyle
{IEEEtran}
\bibliography{IEEEabrv,Cite}

\end{appendices}
\end{document}